\documentclass[11pt,letterpaper]{article}
\usepackage[margin=0.8in]{geometry}
\usepackage{ifpdf}
\ifpdf
    \usepackage[pdftex]{graphicx}
    \graphicspath{{figs/}{matlab/}}   
    \usepackage[update]{epstopdf}
\else
	\usepackage{graphicx}
	\graphicspath{{figs/}{matlab/}}   
\fi

\usepackage{hyperref}
\usepackage{wrapfig,bbm}
\usepackage{enumitem,color}
\usepackage{amsthm}
\newtheorem{theorem}{Theorem}
\newtheorem{lemma}{Lemma}
\newtheorem{corollary}{Corollary}

\newtheorem{prop}{\textbf{Proposition}}
 
\usepackage{array}
\usepackage{amsmath,amsthm}

\usepackage{amssymb}
\usepackage{amstext}
\usepackage{cite,setspace}

\usepackage{breqn}

\usepackage{multirow}

\newtheorem{definition}{\textbf{Definition}}

\allowdisplaybreaks

\begin{document}

\title{Multilevel Diversity Coding with Regeneration}
\author{Chao Tian and Tie Liu
\thanks{This paper was presented in part at the 2014 Annual Allerton Conference on Communications, Control, and Computing and in part at the 2015 IEEE International Symposium on Information Theory. Chao Tian is with the Department of Electrical Engineering and Computer Science, The University of Tennessee, Knoxville, TN 37996, USA (email: chao.tian@utk.edu). C.~Tian's work was supported in part by the National Science Foundation under Grant CCF-15-26095. Tie Liu is with the Department of Electrical and Computer Engineering, Texas A\&M University, College Station, TX 77843, USA (email: tieliu@tamu.edu). T. Liu's work was supported in part by the National Science Foundation under Grants CCF-13-20237 and CCF-15-24839.}} 
\maketitle

\begin{abstract}
Digital contents in large-scale distributed storage systems may have different reliability and access delay requirements, and erasure codes with different strengths can provide the best storage efficiency in these systems. At the same time, in such large-scale distributed storage systems, nodes fail on a regular basis, and the contents stored on them need to be regenerated from the data downloaded from the remaining nodes. The efficiency of this repair process is an important factor that affects the overall quality of service. In this work, we formulate the problem of multilevel diversity coding with regeneration to address these considerations, for which the storage vs. repair-bandwidth tradeoff is investigated. We show that the extreme point on the optimal tradeoff curve that corresponds to the minimum possible storage can be achieved by a simple coding scheme, in which contents with different reliability requirements are encoded separately with individual regenerating codes without any mixing. On the other hand, we establish the complete storage-repair-bandwidth tradeoff for the case of four storage nodes, which reveals that codes mixing different contents can, in general, strictly improve the optimal tradeoff over the separate-coding solution.     
\end{abstract}

\section{Introduction}

The importance of big-data analytics has been widely recognized in recent years. However, efficient large-scale distributed data storage systems have to be designed and implemented in order to support the complete pipeline of data collection, processing and archival on a scale that has never been put into practice before. Advanced coding techniques have been shown to be helpful in terms of providing both performance improvement and cost reduction in such systems.

Digital contents in large-scale distributed storage systems usually have different reliability requirements. For example, although it is important to protect recent customer billing records with a very reliable code, it may be acceptable to allow the data loss probability of a five-year-old office document backup to be higher by using a weaker code. Moreover, erasure codes can also be used to reduce data access queuing delays; see \cite{Walsh:09, Huang:12:ISIT, Shah:14} and references therein. Thus, different levels of latency can also be integrated into the same data storage system by adopting different coding parameters for different contents. Such flexibility can significantly reduce the cost of hardware infrastructure, and there is a tremendous amount of interest recently in both industry and academia to design efficient software-defined storage (SDS) systems utilizing flexible erasure codes. The theoretical framework of symmetrical multilevel diversity (MLD) coding \cite{RocheYeungHau:97,YeungZhang:99} is a natural fit for this scenario, where a total of $k_0$ independent messages $(M_1,M_2,...,M_{k_0})$ are to be stored in $n \geq k_0$ storage nodes situated in different network locations, each with $\alpha$ units of data. The messages are coded in such a way that by accessing any $k\leq k_0$ of these nodes, the first $k$ messages $(M_1,M_2,...,M_k)$ can be completely recovered. 

Disk or node failures occur regularly in a large-scale data storage system, and the overall quality of service is heavily affected by the efficiency of the repair process. Dimakis {\em et al.} \cite{Dimakis:10} proposed the framework of regenerating codes to address the tradeoff between the storage and repair-bandwidth in $(n,k)$ erasure-code-based distributed storage systems. To repair a node, a new node replacing the failed one requests $\beta$ units of data each from any of the $d$ remaining nodes, and regenerates the $\alpha$ units of content to store on the new node; this code is referred to as an  $(n,k,d)$ regenerating code. There exists a natural tradeoff between the storage $\alpha$ and the repair bandwidth $\beta$: The point corresponding to the minimum amount of the storage is referred to as the minimum storage regenerating (MSR) point, and the other extreme corresponding to the minimum amount of repair bandwidth is referred to as the minimum repair-bandwidth regenerating (MBR) point. In \cite{Dimakis:10}, the content regenerated is allowed to be only functionally equivalent to the original content stored on the failed node, thus the name \lq\lq{}functional-repair\rq\rq{} regenerating codes. In practice, requiring the content regenerated to be exactly the same as that stored on the failed node can simplify the system design significantly, and thus recent research effort has been focusing on \lq\lq{}exact-repair\rq\rq{} regenerating codes \cite{Dimakis:11,RashmiShah:11, RashmiShah:12:1,Cadambe:11,Tamo:13,Papailiopoulos:13,Tian:JSAC13}.

In the current regenerating code framework, only a single message is allowed, and thus only a single level of reliability and access latency is offered. On the other hand, in the classical MLD coding framework, the data repair process was not considered. In this work, we consider repair-efficient codes in systems with heterogeneous reliability and latency requirements, and investigate the optimal storage vs. repair-bandwidth tradeoff. Because of the connection to the MLD coding and regenerating code problems, we refer to this problem as multilevel diversity coding with regeneration (MLD-R) in the sequel (see Fig. \ref{fig:system} for an illustration of the system). We shall restrict our attention to the case of exact-repair and, furthermore, to the case when $d=n-1$, because this is the most practically important case. Nevertheless, the proposed framework can be generalized to other relevant settings in straightforward fashion. 

An intuitive and straightforward coding strategy for MLD-R is to use an individual regenerating code for each message to satisfy the respective reliability and latency requirement ({\em i.e.,} separate coding), and thus an important question that we wish to answer first is whether it is even beneficial to consider codes that \lq\lq{}mix\rq\rq{} the messages. Without the repair consideration, it was shown in \cite{RocheYeungHau:97,YeungZhang:99} that mixing is not necessary for the (symmetrical) MLD coding problem. As we shall see shortly, for the minimum storage point on the optimal tradeoff curve where $\alpha$ is minimized (analogous to the MSR point in standard regenerating codes), the aforementioned separate-coding strategy is again sufficient. On the other hand, we show for $n=4$, by providing a novel code construction, that mixing can strictly improve upon the performance of the separate-coding solution in terms of the overall storage-repair-bandwidth tradeoff. In fact, we completely characterize the optimal tradeoff for this case by establishing its converse. It is worth noting that when $n=3$, separate coding is sufficient, thus $n=4$ is the smallest non-trivial case where the benefit of mixing manifests. 

The main difficulty for establishing the aforementioned results is in deriving the tight outer bounds. For the minimum storage point, we utilize a recursive bounding technique which may be of independent interest. The converse for the tradeoff rate region when $n=4$ is rather difficult to identify and derive analytically, and our approach is to utilize the computational method developed in \cite{Tian:JSAC13}. The proof is thus presented in tables whose rows are simple known information inequalities, and the summations of the rows give precisely the desired outer bounds. Though this does not conform to the conventional approach of using chains of information inequalities in information theory literature, we believe that these tables are, in fact, more fundamental: We can write down many different versions of chains of inequalities with the help of these tables, by taking different orders when applying these individual inequalities.       

The rest of the paper is organized as follows. A formal problem formulation and some preliminaries are given in Section \ref{sec:problemformulation}. In Section \ref{sec:main}, the main results of the paper are presented together with the relevant discussions. The proofs are given in Sections \ref{sec:msp} and \ref{sec:case4}. Section \ref{sec:conclusion} concludes the paper with a few possible future research directions. Several technical proofs are given in the Appendix. 

\begin{figure}
\centering
\includegraphics[width=17cm]{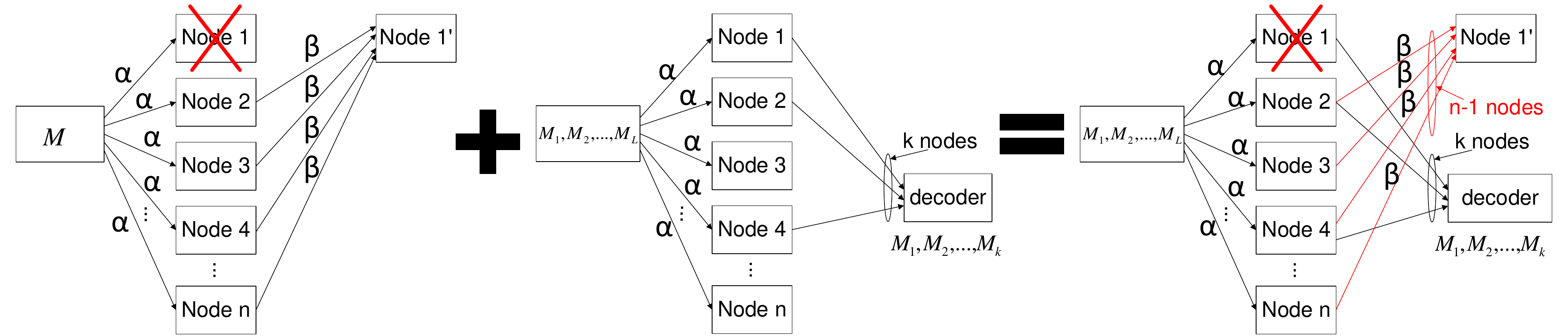}
\caption{Multilevel diversity coding with regeneration. \label{fig:system}}
\end{figure}

\section{Problem Formulation and Preliminaries}
\label{sec:problemformulation}

\subsection{Problem Formulation}

An MLD-R code is formally defined below, where $I_n$ denotes the set $\{1,2,\ldots,n\}$ and $|A|$ denotes the cardinality of a set $A$. Without loss of generality, we may assume that the number of nodes accessed during repair, {\em i.e.,} the parameter $d$, is the same as $k_0$, where $k_0$ is the number of messages. This is because if $d<k_0$, then the messages\footnote{For readers familiar with \cite{RocheYeungHau:97,YeungZhang:99}, the messages here correspond to the independent sources in \cite{RocheYeungHau:97,YeungZhang:99}. Our problem can be alternatively defined using such sources at the expense of more sophisticated notations.} $(M_{d+1},M_{d+2},...,M_{k_0})$ can be viewed as part of $M_d$ as they can all be reconstructed by accessing any $d$ nodes. On the other hand, if $d>k_0$, we can simply consider an alternative problem with $k\rq{}_0=d$ and define the messages  $M_{k_0},M_{k_0+1},...,M_{d}$ to be degenerate ({\em i.e.,} with rate zero). Recall that we shall assume $d=n-1$ for the rest of the paper.

\begin{definition}
\label{def:NKKcode}
An $(N_1,N_2,...,N_d,K_d,K)$ MLD-R code consists of $n$ encoding functions $f^E_i(\cdot)$,  $\sum_{i=1}^d{n \choose i}$ decoding functions $f^D_{A}(\cdot,...,\cdot)$, $n(n-1)$ repair-encoding functions $F^{E}_{i,j}(\cdot)$,  and $n$ repair-decoding functions $F^{D}_{j}(\cdot,...,\cdot)$, where
\begin{itemize}
\item
$f^E_i:I_{N_1}\times I_{N_2}\times...\times I_{N_d}\rightarrow I_{K_d}$ for $i\in I_n$, 
each of which maps the messages $(M_1,M_2,...,M_d)\in I_{N_1}\times I_{N_2}\times...\times I_{N_d}$ to one piece of coded information to be stored on node $i$;
\item $f^D_{A}:I_{K_d}\times I_{K_d}\times...\times I_{K_d}\rightarrow I_{N_1}\times I_{N_2}\times...\times I_{N_{|A|}}$ for $ A\subset {I}_n$ and $|A|=1,2,\ldots,d$, 
each of which maps the coded information stored on a set $A$ of nodes to the first $|A|$ messages $(M_1,M_2,...,M_{|A|})$;
\item 
$F^{E}_{i,j}:I_{K_d}\rightarrow I_{K}$ for $j\in I_n$ and $i\in I_n\setminus \{j\}$, 
each of which maps a piece of coded information at node $i$ to an index that is made available to regenerate the coded data stored at node $j$; and
\item
$F^{D}_{j}:{I}_{K}\times{I}_{K}\times...\times{I}_{K} \rightarrow {I}_{K_d}$ for $j\in{I}_n$, 
each of which maps $d$ such indices from the helper nodes $I_n\setminus\{j\}$ to regenerate the information stored at the failed node $j$.
\end{itemize}
The functions must satisfy:
\begin{itemize}
\item[1)] the data-reconstruction conditions
\begin{align}
&f_{{A}}^D\left( f^E_i(M_1,M_2,...,M_{d}),i\in{A}\right)=(M_1,M_2,...,M_{|A|}),\nonumber\\
&\qquad (M_1,M_2,...,M_d)\in{I}_{N_1}\times{I}_{N_2}\times...\times {I}_{N_d},\nonumber\\
&\qquad\qquad \qquad\qquad\qquad{A}\subset {I}_n\,\, \mbox{and}\,\, |{A}|=1,2,\ldots,d
\label{eqn:reconstructionzeroerror}
\end{align}
\item[2)] and the node-regeneration conditions
\begin{align}\label{eqn:repairrequirement}
&F^D_{j}\left(F^E_{i,j}\left(f^E_i(M_1,M_2,...,M_d)\right),i\in{I_n\setminus \{j\}}\right)=f^E_j(M_1,M_2,...,M_d),\nonumber\\
&\qquad\quad (M_1,M_2,...,M_d)\in{I}_{N_1}\times{I}_{N_2}\times...\times {I}_{N_d}\,\, \mbox{and} \,\, j\in{I}_n.
\end{align}
\end{itemize}
\end{definition}

Note that $\alpha=\log K_d$ is the storage node capacity,  $\beta=\log K$ is the per-helper-node repair bandwidth, and $B_i=\log N_i$ is the rate of the $i$-th message. The base of $\log(\cdot)$ is arbitrary, and we choose base 2 for convenience. As a concrete example, consider the case with $n=3$ nodes. Each of the three nodes has a storage capacity $\alpha$. There are two messages $M_1$ and $M_2$, the first of which needs to be reconstructed by accessing any one node, and the latter of which needs to be reconstructed by accessing any two nodes. Any single node failure needs to be repairable by using the remaining two nodes, each of which contributes $\beta$ amount of helper data. Because of the linear-scaling relation among them, we can alternatively consider the normalized version of $\alpha$, $\beta$ and $B_i$ as follows. 

\begin{definition}
A normalized storage-repair-bandwidth-message-rate tuple $(\bar{\alpha},\bar{\beta},\bar{B}_1,\bar{B}_2,...,\bar{B}_d)$ is said to be achievable with $n$ nodes where $\sum_{j=1}^d{\bar{B}_j}=1$, if there exists an $(N_1,N_2,...,N_d,K_d,K)$ MLD-R code such that
\begin{align}
&\bar{\alpha}\geq \frac{\log K_d}{\sum_{i=1}^d\log N_i},\,\,
\bar{\beta}\geq \frac{\log K}{\sum_{i=1}^d\log N_i}\,\,\mbox{and}\,\,\bar{B}_j=\frac{\log N_i}{\sum_{i=1}^d\log N_i},\qquad  j=1,2,...,d.\nonumber
\end{align}
The closure of all achievable $(\bar{\alpha},\bar{\beta},\bar{B}_1,\bar{B}_2,...,\bar{B}_d)$ tuples is the achievable normalized storage-repair-bandwidth-message-rate tradeoff region $\mathcal{R}_n$. For a fixed $(\bar{B}_1,\bar{B}_2,...,\bar{B}_d)$ tuple, the achievable normalized storage-repair-bandwidth tradeoff region is the collection of all $(\bar{\alpha},\bar{\beta})$ pairs such that $(\bar{\alpha},\bar{\beta},\bar{B}_1,\bar{B}_2,...,\bar{B}_d)\in \mathcal{R}_n$, which is denoted as $\mathcal{R}_n(\bar{B}_1,\bar{B}_2,...,\bar{B}_d)$.
\end{definition}

The codes and the tradeoff regions do not involve any particular assumption on the distribution of the messages. However, without loss of generality we may assume that the messages $M_1,M_2,\ldots, M_d$ are mutually independent and uniformly distributed, since otherwise we can perform a pre-coding to eliminate any dependency and non-uniformity. 
Sometimes it is convenient to use the accumulative sum rates instead of the individual rates, and we thus define
\begin{align}
\bar{B}^+_k\triangleq\sum_{i=1}^k\bar{B}_i,\quad k=1,2,\ldots,d.
\end{align}
Note that this definition implies $\bar{B}^+_d=1$, even though we often still write $\bar{B}^+_d$ for convenience. 

The data-reconstruction condition (\ref{eqn:reconstructionzeroerror}) requires that there is no decoding error, {\em i.e.}, the zero-error requirement is adopted. An alternative definition is to require, instead, the probability of decoding error to vanish in the limit as $\Pi_{i=1}^d N_i\rightarrow \infty$. It will become clear that this does not cause any essential difference, and we thus do not further discuss this alternative definition in this paper. 

When deriving outer bounds, we use $S_{i\rightarrow j}$ to denote the random variable representing the helper data sent from node $i$ to node $j$ during the repair of node $j$, {\em i.e.}, the output of the function $f^E_{i,j}$, and $W_i$ to denote the random variable representing the coded data stored on node $i$, {\em i.e.}, the output of the function $f^E_i$. The random vector $(X_1,X_2,\ldots,X_m)$ is sometimes written as $X_1^m$ for notational simplicity.

\subsection{Separate Coding}

One straightforward coding strategy is to encode each individual message separately using a regenerating code of the necessary parameters. More precisely, suppose that each message $M_k$ is encoded using an $(n,k,d=n-1)$ regenerating code ({\em i.e.}, any $k$ nodes can recover the message $M_k$, and any new node obtains data from any $d$ nodes for repair) of rate $(\alpha_k,\beta_k)$. Then, the resulting code has storage and repair rates given by
\begin{align*}
\alpha = \sum_{k=1}^d \alpha_k\quad \mbox{and}  \quad \beta=\sum_{k=1}^d \beta_k
\end{align*}
respectively. Therefore, if we assume that the normalized rate pair $(\bar{\alpha}_k,\bar{\beta}_k)$ is achievable by an $(n,k,d=n-1)$ regenerating code, the normalized rate pair
\begin{align}
(\bar{\alpha},\bar{\beta})=\left(\sum_{k=1}^{d}\bar{\alpha}_k\bar{B}_k,\sum_{k=1}^{d}\bar{\beta}_k\bar{B}_k\right)\label{eq:SC}
\end{align}
is achievable by separate encoding. The collection of all normalized rate pairs \eqref{eq:SC}, over all achievable normalized rate pairs $(\bar{\alpha}_k,\bar{\beta}_k)$ for any individual $(n,k,d=n-1)$ regenerating code, is the separate-coding normalized tradeoff region and is denoted as $\hat{\mathcal{R}}_{n}(\bar{B}_1,\bar{B}_2,\ldots,\bar{B}_d)$.

In order to characterize the separate-coding normalized tradeoff region $\hat{\mathcal{R}}_n(\bar{B}_1,\bar{B}_2,...,\bar{B}_d)$, normalized tradeoff region characterizations of individual regenerating codes are needed. For example, for the case of $n=3$, normalized tradeoff region characterizations for $(3,1,2)$ and $(3,2,2)$ regenerating codes are needed. However, such characterizations for general parameters are still unknown, except for the case of $k=1,2$ (for an arbitrary $d$ and $n$), and the special case $(n,k,d)=(4,3,3)$ recently established in \cite{Tian:JSAC13}. Fortunately, using these existing results, we can provide precise characterizations of the separate-coding normalized tradeoff regions for MLD-R when $n=3,4$.

\begin{lemma}
\label{prop:case3}
The separate-coding normalized tradeoff region $\hat{\mathcal{R}}_{3}(\bar{B}_1,\bar{B}_2)$ is the set of $(\bar{\alpha},\bar{\beta})$ pairs satisfying the following conditions:
\begin{align}
\bar{\alpha}\geq \bar{B}_1+\frac{\bar{B}_2}{2},\quad
\bar{\alpha}+\bar{\beta}\geq \frac{3\bar{B}_1}{2}+\bar{B}_2,\quad \mbox{and} \quad
\bar{\beta}\geq \frac{\bar{B}_1}{2}+\frac{\bar{B}_2}{3}.\label{eq:case3}
\end{align}  
\end{lemma}

\begin{lemma}
\label{prop:case4}
The separate-coding normalized tradeoff region $\hat{\mathcal{R}}_{4}(\bar{B}_1,\bar{B}_2,\bar{B}_3)$ is the set of $(\bar{\alpha},\bar{\beta})$ pairs satisfying the following conditions:
\begin{align}
\bar{\alpha} &\geq \bar{B}_1+\frac{\bar{B}_2}{2}+\frac{\bar{B}_3}{3},\label{eqn:case4_1}\\
2\bar{\alpha}+\bar{\beta} &\geq \frac{7\bar{B}_1}{3}+\frac{5\bar{B}_2}{4}+\bar{B}_3,\label{eqn:case4_2}\\
4\bar{\alpha}+6\bar{\beta} &\geq 6\bar{B}_1+\frac{7\bar{B}_2}{2}+3\bar{B}_3,\label{eqn:case4_3}\\
\bar{\alpha}+2\bar{\beta} &\geq \frac{5\bar{B}_1}{3}+\bar{B}_2+\frac{5}{6}\bar{B}_3,\label{eqn:case4_4}\\
\mbox{and} \quad \bar{\beta} &\geq \frac{\bar{B}_1}{3}+\frac{\bar{B}_2}{5}+\frac{\bar{B}_3}{6}.\label{eqn:case4_5}
\end{align}  
\end{lemma}

The proofs of these lemmas are given in the Appendix.

\section{Main Results}
\label{sec:main}

Our first main result is a precise characterization of the extreme point in the tradeoff rate region $\mathcal{R}_n(\bar{B}_1,\bar{B}_2,...,\bar{B}_d)$ where $\bar{\alpha}$ is minimized.

\begin{theorem}\label{theorem:msp}
For any $(\bar{\alpha},\bar{\beta})\in\mathcal{R}_n(\bar{B}_1,\bar{B}_2,...,\bar{B}_d)$, we have
\begin{align}
(n-2)\bar{\alpha}+\bar{\beta}\geq \sum_{k=1}^{n-1}\frac{(n-2)(n-k)+1}{k(n-k)}\bar{B}_k. \label{eqn:bound1}
\end{align}
Moreover, the minimum storage point of $(\bar{\alpha},\bar{\beta})\in\mathcal{R}_n(\bar{B}_1,\bar{B}_2,...,\bar{B}_d)$ is given as 
\begin{align}
(\bar{\alpha},\bar{\beta})=\left(\sum_{k=1}^{n-1}\frac{\bar{B}_k}{k},\sum_{k=1}^{n-1}\frac{\bar{B}_k}{k(n-k)}\right)\label{eqn:MSRalphabeta}
\end{align}
which can be achieved by separately coding of each message $M_k$ with an $(n,k,d=n-1)$ exact-repair MSR code.
\end{theorem}

The outer bound (\ref{eqn:bound1}) is proved in the Appendix. To see the second part of the theorem, recall that for MLD coding without the repair consideration, the minimum normalized storage rate is given in \cite{RocheYeungHau:97} as 
\begin{align}
\bar{\alpha}=\sum_{k=1}^{n-1}\frac{\bar{B}_k}{k}. \label{eqn:miniratealpha}
\end{align}
Note that this minimum normalized storage rate is also achievable with the additional repair consideration, as the entire collection of the messages $(M_1,M_2,\ldots,M_d)$ can be recovered by downloading the data stored at any $d$ storage nodes. Plugging the minimum normalized storage rate \eqref{eqn:miniratealpha} into (\ref{eqn:bound1}) gives:
\begin{align}
\bar{\beta} \geq \sum_{k=1}^{n-1}\frac{\bar{B}_k}{k(n-k)}.
\end{align}
On the other hand, for each $(n,k,d=n-1)$ exact-repair MSR code, the following MSR point is known to be achievable \cite{Cadambe:11}:
\begin{align}
(\bar{\alpha}_k,\bar{\beta}_k)=\left(\frac{1}{k},\frac{1}{k(n-k)}\right).
\end{align}
Thus the minimum storage point \eqref{eqn:MSRalphabeta} can be achieved by the aforementioned separate-coding strategy, when we let each individual exact-repair regenerating code operate at their respective MSR points.

Our next two results provide complete characterizations for $\mathcal{R}_{3}(\bar{B}_1,\bar{B}_2)$ and $\mathcal{R}_{4}(\bar{B}_1,\bar{B}_2,\bar{B}_3)$.
\begin{theorem}
\label{theorem:case3}
$\mathcal{R}_{3}(\bar{B}_1,\bar{B}_2)=\hat{\mathcal{R}}_{3}(\bar{B}_1,\bar{B}_2)$. 
\end{theorem}

We obviously have $\hat{\mathcal{R}}_{3}(\bar{B}_1,\bar{B}_2)\subseteq\mathcal{R}_{3}(\bar{B}_1,\bar{B}_2)$, and the reversed inclusion $\mathcal{R}_{3}(\bar{B}_1,\bar{B}_2)\subseteq \hat{\mathcal{R}}_{3}(\bar{B}_1,\bar{B}_2)$ is proved in the Appendix. The theorem states that for the case of $n=3$, the separate-coding strategy is optimal, and there is no need to mix the messages. However, our next result shows that this is in general not the case and mixing the messages can be (strictly) beneficial.

\begin{theorem}\label{theorem:MRR4}
The normalized storage-repair-bandwidth tradeoff region  $\mathcal{R}_4(\bar{B}_1,\bar{B}_2,\bar{B}_3)$ is the collection of all $(\bar{\alpha},\bar{\beta})$ pairs that satisfy the following constraints:
\begin{align}
\bar{\alpha}  &\geq\bar{B}_1+\frac{1}{2}\bar{B}_2+\frac{1}{3}\bar{B}_3,\label{eqn:rateregion4_1}\\
2\bar{\alpha}+\bar{\beta}&\geq \frac{7}{3}\bar{B}_1+\frac{5}{4}\bar{B}_2+\bar{B}_3,\label{eqn:rateregion4_2}\\
\bar{\alpha} +\bar{\beta}&\geq \frac{4}{3}\bar{B}_1+\frac{3}{4}\bar{B}_2+\frac{5}{8}\bar{B}_3,\label{eqn:rateregion4_3}\\
2\bar{\alpha}+3\bar{\beta}&\geq 3\bar{B}_1+\frac{5}{3}\bar{B}_2+\frac{3}{2}\bar{B}_3,\label{eqn:rateregion4_4}\\
\bar{\alpha} +2\bar{\beta}&\geq \frac{5}{3}\bar{B}_1+\bar{B}_2+\frac{5}{6}\bar{B}_3,\label{eqn:rateregion4_5}\\
\mbox{and} \quad \bar{\beta}&\geq \frac{1}{3}\bar{B}_1+\frac{1}{5}\bar{B}_2+\frac{1}{6}\bar{B}_3\label{eqn:rateregion4_6}.
\end{align}
\end{theorem}

The proof of the theorem is provided in Section~\ref{sec:case4}. The regions $\mathcal{R}_4\left(0,\frac{1}{3},\frac{2}{3}\right)$ and $\hat{\mathcal{R}}_4\left(0,\frac{1}{3},\frac{2}{3}\right)$ are depicted in Fig. \ref{fig:rateregion}. It can be seen that the inclusion $\hat{\mathcal{R}}_4\left(0,\frac{1}{3},\frac{2}{3}\right)\subseteq\mathcal{R}_4\left(0,\frac{1}{3},\frac{2}{3}\right)$ is strict, thus in general mixing of contents in MLD-R can be beneficial. We formally state this fact next. 

\begin{corollary}
\label{coro:different}
$\hat{\mathcal{R}}_4(\bar{B}_1,\bar{B}_2,\bar{B}_3)\subsetneq\mathcal{R}_4(\bar{B}_1,\bar{B}_2,\bar{B}_3)$ if and only if $\bar{B}_2\bar{B}_3>0$.
\end{corollary}

Since the general forms of $\hat{\mathcal{R}}_4(\bar{B}_1,\bar{B}_2,\bar{B}_3)$ and $\mathcal{R}_4(\bar{B}_1,\bar{B}_2,\bar{B}_3)$ are different, it is expected that they are not identical except for certain degenerate cases. These degenerate cases turn out to be precisely when either $\bar{B}_2=0$ or $\bar{B}_3=0$. This corollary is proved in the Appendix.

\begin{figure}
\begin{centering}
\includegraphics[width=11cm]{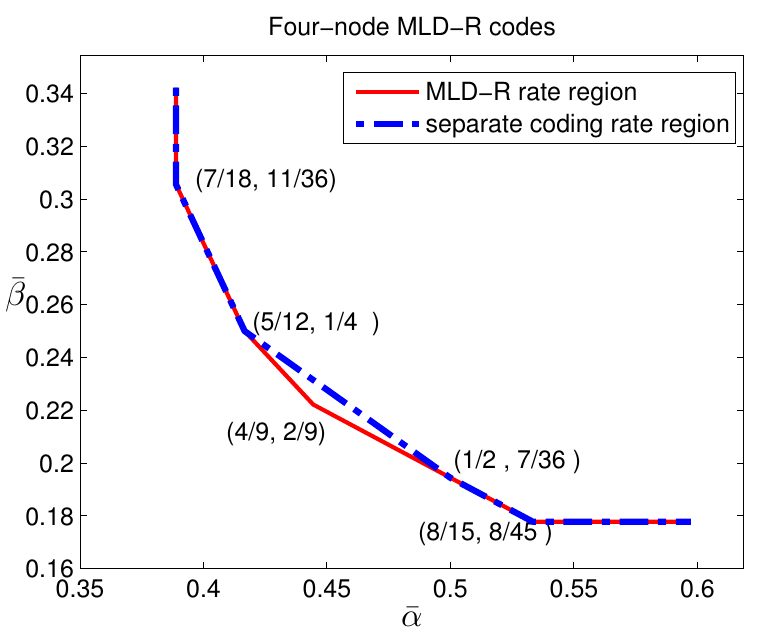}
\caption{The normalized separate coding rate region $\hat{\mathcal{R}}_4\left(0,\frac{1}{3},\frac{2}{3}\right)$ and the normalized tradeoff rate region ${\mathcal{R}}_4\left(0,\frac{1}{3},\frac{2}{3}\right)$. \label{fig:rateregion}}
\end{centering}
\end{figure}

\section{The Minimum Storage Point: Proof of Theorem \ref{theorem:msp}}
\label{sec:msp}

It can be shown that we only need to consider symmetric codes, where permutations of node indices do not 
change the induced joint entropy values. See \cite{Tian:JSAC13} for more details about this type of symmetry. Thus, without loss of generality we may restrict the proof to symmetric codes only. Before presenting the proof of Theorem \ref{theorem:msp}, we shall first present an auxiliary lemma, which will play an important role in the induction proof of Theorem \ref{theorem:msp}.  The proof of the lemma makes use of the celebrated Han's inequality \cite{Han-IC78}, partly motivated by the converse proof of (symmetrical) MLD coding problem (without the regeneration requirement) given in \cite{RocheYeungHau:97}.
\begin{lemma}\label{lemma:msp}
For any integer $\ell$ such that $1\leq \ell\leq d-1$ and any symmetric MLD-R code with a total of $n$ nodes, we have
\begin{align}
&\frac{n-1}{\ell(n-\ell)}H(W^{\ell}_1|M^{\ell}_1)+\frac{n-\ell-1}{n-\ell}H(S_{\ell+1\rightarrow1},\ldots,S_{n\rightarrow1},W^{\ell}_1|M^{\ell}_1)\nonumber\\
&\qquad\qquad\geq B_{\ell+1}+H(S_{\ell+2\rightarrow1},\ldots,S_{n\rightarrow1},W^{\ell+1}_1|M^{\ell+1}_1).
\end{align}
\end{lemma}
\begin{proof}
Let $\ell$ be an integer such that $1\leq \ell\leq d-1$. We start by writing
\begin{align}
&H(W^{\ell}_1|M^{\ell}_1)+H(S_{\ell+1\rightarrow1},\ldots,S_{n\rightarrow1},W^{\ell}_1|M^{\ell}_1)\nonumber\\
&\stackrel{(s)}{=}H(W^{\ell+1}_2|M^{\ell}_1)+H(S_{\ell+1\rightarrow1},\ldots,S_{n\rightarrow1},W^{\ell}_1|M^{\ell}_1)\nonumber\\
&\stackrel{(a)}{=}H(S_{\ell+1\rightarrow 1},W^{\ell+1}_2|M^{\ell}_1)+H(S_{\ell+1\rightarrow1},\ldots,S_{n\rightarrow1},W^{\ell}_1|M^{\ell}_1)\nonumber\\
&\stackrel{(b)}{\geq} H(S_{\ell+2\rightarrow1},\ldots,S_{n\rightarrow1},W^{\ell+1}_1|M^{\ell}_1)+H(S_{\ell+1\rightarrow 1},W^{\ell}_2|M^{\ell}_1)\nonumber\\
&\stackrel{(c)}{=} H(S_{\ell+2\rightarrow1},\ldots,S_{n\rightarrow1},W^{\ell+1}_1,M_{\ell+1}|M^{\ell}_1)+H(S_{\ell+1\rightarrow 1},W^{\ell}_2|M^{\ell}_1)\nonumber\\
&\stackrel{(d)}{=} B_{\ell+1}+H(S_{\ell+2\rightarrow1},\ldots,S_{n\rightarrow1},W^{\ell+1}_1|M^{\ell+1}_1)+H(S_{\ell+1\rightarrow 1},W^{\ell}_2|M^{\ell}_1) \label{eqn:lemma1}
\end{align}
where we write, from now on, $(s)$ to denote ``the reason of symmetry''. Here, $(a)$ is due to the fact that $S_{\ell+1\rightarrow 1}$ is a function of $W_{\ell+1}$, $(b)$ follows from the standard submodularity of the entropy function, $(c)$ is due to the fact that $M_{\ell+1}$ can be recovered from $W_{1}^{\ell+1}$, and $(d)$ follows from the assumption that $M_{\ell+1}$ is independent of $M_1,\ldots,M_\ell$  and uniform over $I_{N_{\ell+1}}$.
Further notice that 
\begin{align}
(n-\ell)H(S_{\ell+1\rightarrow 1},W^{\ell}_2|M^{\ell}_1)
&=(n-\ell)H(W^{\ell}_2|M^{\ell}_1)+(n-\ell)H(S_{\ell+1\rightarrow 1}|M^{\ell}_1,W^{\ell}_2)\nonumber\\
&\stackrel{(s)}{=}(n-\ell)H(W^{\ell}_2|M^{\ell}_1)+\sum_{i=\ell+1}^nH(S_{i\rightarrow 1}|M^{\ell}_1,W^{\ell}_2)\nonumber\\
&\geq (n-\ell)H(W^{\ell}_2|M^{\ell}_1)+H(S_{\ell+1\rightarrow 1},\ldots,S_{n\rightarrow 1}|M^{\ell}_1,W^{\ell}_2)\nonumber\\
&=(n-\ell)H(W^{\ell}_2|M^{\ell}_1)+H(S_{\ell+1\rightarrow 1},\ldots,S_{n\rightarrow 1},W_1|M^{\ell}_1,W^{\ell}_2)\nonumber\\
&=(n-\ell-1)H(W^{\ell}_2|M^{\ell}_1)+H(S_{\ell+1\rightarrow 1},\ldots,S_{n\rightarrow 1},W^{\ell}_1|M^{\ell}_1)\nonumber\\
&\geq \frac{(n-\ell-1)(\ell-1)}{\ell}H(W^{\ell}_1|M^{\ell}_1)+H(S_{\ell+1\rightarrow 1},\ldots,S_{n\rightarrow 1},W^{\ell}_1|M^{\ell}_1) \label{eqn:lemma2}
\end{align}
where in the last step we applied the conditional version of Han's inequality \cite{Han-IC78} under the symmetry assumption:
\begin{align}
\frac{1}{l-1}H(W^{\ell}_2|M^{\ell}_1) \geq \frac{1}{l}H(W^{\ell}_1|M^{\ell}_1).
\end{align}
Putting (\ref{eqn:lemma1}) and (\ref{eqn:lemma2}) together gives
\begin{align}
&H(W^{\ell}_1|M^{\ell}_1)+H(S_{\ell+1\rightarrow1},\ldots,S_{n\rightarrow1},W^{\ell}_1|M^{\ell}_1)\nonumber\\
&\geq B_{\ell+1}+H(S_{\ell+2\rightarrow1},\ldots,S_{n\rightarrow1},W^{\ell+1}_1|M^{\ell+1}_1)+\frac{(n-\ell-1)(\ell-1)}{\ell(n-\ell)}H(W^{\ell}_1|M^{\ell}_1)\nonumber\\
&\qquad\qquad+\frac{H(S_{\ell+1\rightarrow 1},\ldots,S_{n\rightarrow 1},W^{\ell}_1|M^{\ell}_1)}{n-\ell},
\end{align}
which has common terms involving $H(W^{\ell}_1|M^{\ell}_1)$ and $H(S_{\ell+1\rightarrow1},\ldots,S_{n\rightarrow1},W^{\ell}_1|M^{\ell}_1)$ on both sides that can be eliminated, and this leads to exactly the inequality stated in the lemma.
\end{proof}

In the proof of the lemma above, after several steps of derivation, the same terms in the original quantity reappear, albeit with different coefficients. These terms are subtracted on both sides of the inequality, which can be conceptually viewed as recursively applying the same chains of inequalities. We are now ready to present the proof of Theorem \ref{theorem:msp}.

\begin{proof}[Proof of Theorem \ref{theorem:msp}]
The theorem is proved through an induction, where we show that for $m=1,2,\ldots,d$,
\begin{align}
&(n-2)\alpha+\beta \geq \sum_{k=1}^{m}\frac{(n-2)(n-k)+1}{k(n-k)}{B}_k+ \left[\frac{n-2}{m}-\frac{m-1}{m(n-m)}\right]H(W^m_1|M^m_1)\nonumber\\
&\qquad\qquad\qquad\qquad\qquad+\frac{H(S_{m+1\rightarrow1},\ldots,S_{n\rightarrow1},W^m_1|M^m_1)}{n-m}.\label{eqn:claim}
\end{align}
The theorem is then simply a consequence of this statement when setting $m=d=n-1$, normalizing both sides by $\sum_{k=1}^{n-1}B_k$, and taking into account of the facts that
\begin{align}
\left[\frac{n-2}{m}-\frac{m-1}{m(n-m)}\right]=0
\end{align}
for $m=n-1$ and
\begin{align}
H(S_{m+1\rightarrow1},\ldots,S_{n\rightarrow1},W^m_1|M^m_1)\geq 0. 
\end{align}

To show that (\ref{eqn:claim}) is true for $m=1$, we write the following chain of inequalities:
\begin{align}
(n-2)\alpha+\beta
&\geq (n-2)H(W_1)+H(S_{2\rightarrow1})\nonumber\\
&= (n-2)H(W_1)+\frac{(n-1)H(S_{2\rightarrow1})}{n-1}\nonumber\\
&\geq (n-2)H(W_1)+\frac{H(S_{2\rightarrow1},S_{3\rightarrow1},\ldots,S_{n\rightarrow1})}{n-1}\nonumber\\
&= (n-2)B_1+(n-2)H(W_1|M_1)+\frac{H(S_{2\rightarrow1},S_{3\rightarrow1},\ldots,S_{n\rightarrow1},W_1)}{n-1}\nonumber\\
&=(n-2)B_1+(n-2)H(W_1|M_1)+\frac{B_1}{n-1}+\frac{H(S_{2\rightarrow1},S_{3\rightarrow1},\ldots,S_{n\rightarrow1},W_1|M_1)}{n-1}\nonumber\\
&=\frac{(n-2)(n-1)+1}{n-1}B_1+(n-2)H(W_1|M_1)+\frac{H(S_{2\rightarrow1},S_{3\rightarrow1},\ldots,S_{n\rightarrow1},W_1|M_1)}{n-1}.
\end{align}
Thus (\ref{eqn:claim}) is true for $m=1$. 

Next suppose that (\ref{eqn:claim}) is true for $m=m_0$ for some $m_0 \leq d-1$, and we wish to show that it is also true for $m=m_0+1$. Notice that by Lemma \ref{lemma:msp} with $\ell=m_0$, we have
\begin{align}
&\frac{n-1}{m_0(n-m_0)(n-m_0-1)}H(W^{m_0}_1|M^{m_0}_1)+\frac{H(S_{m_0+1\rightarrow1},\ldots,S_{n\rightarrow1},W^{m_0}_1|M^{m_0}_1)}{n-m_0}\nonumber\\
&\qquad\qquad\geq \frac{1}{n-m_0-1}B_{m_0+1}+\frac{1}{n-m_0-1}H(S_{m_0+2\rightarrow1},\ldots,S_{n\rightarrow1},W^{m_0+1}_1|M^{m_0+1}_1).\label{eq:msp}
\end{align}
It thus follows from the induction assumption that
\begin{align}
&(n-2)\alpha+\beta \geq \sum_{k=1}^{m_0}\frac{(n-2)(n-k)+1}{k(n-k)}{B}_k+\left[\frac{n-2}{m_0}-\frac{m_0-1}{m_0(n-m_0)}\right]H(W^{m_0}_1|M^{m_0}_1)\nonumber\\
&\qquad\qquad\qquad\qquad\qquad+\frac{H(S_{m_0+1\rightarrow1},\ldots,S_{n\rightarrow1},W^{m_0}_1|M^{m_0}_1)}{n-m_0}\nonumber\\
&\stackrel{(a)}{\geq} \sum_{k=1}^{m_0}\frac{(n-2)(n-k)+1}{k(n-k)}{B}_k+\left[\frac{n-2}{m_0}-\frac{m_0-1}{m_0(n-m_0)}-\frac{n-1}{m_0(n-m_0)(n-m_0-1)}\right]H(W^{m_0}_1|M^{m_0}_1)\nonumber\\
&\qquad +\frac{1}{n-m_0-1}B_{m_0+1}+\frac{1}{n-m_0-1}H(S_{m_0+2\rightarrow1},\ldots,S_{n\rightarrow1},W^{m_0+1}_1|M^{m_0+1}_1)\nonumber\\
&=\sum_{k=1}^{m_0}\frac{(n-2)(n-k)+1}{k(n-k)}{B}_k+\left[\frac{n-2}{m_0}-\frac{1}{n-m_0-1}\right]H(W^{m_0}_1|M^{m_0}_1)\nonumber\\
&\qquad +\frac{1}{n-m_0-1}B_{m_0+1}+\frac{1}{n-m_0-1}H(S_{m_0+2\rightarrow1},\ldots,S_{n\rightarrow1},W^{m_0+1}_1|M^{m_0+1}_1)\nonumber\\
&\stackrel{(b)}{\geq} \sum_{k=1}^{m_0}\frac{(n-2)(n-k)+1}{k(n-k)}{B}_k+\left[\frac{n-2}{m_0}-\frac{1}{n-m_0-1}\right]\frac{m_0}{m_0+1}H(W^{m_0+1}_1|M^{m_0}_1)\nonumber\\
&\qquad +\frac{1}{n-m_0-1}B_{m_0+1}+\frac{1}{n-m_0-1}H(S_{m_0+2\rightarrow1},\ldots,S_{n\rightarrow1},W^{m_0+1}_1|M^{m_0+1}_1)\nonumber\\
&\stackrel{(c)}{=} \sum_{k=1}^{m_0}\frac{(n-2)(n-k)+1}{k(n-k)}{B}_k+\left[\frac{n-2}{m_0+1}-\frac{m_0}{(n-m_0-1)(m_0+1)}\right]\left[B_{m_0+1}+H(W^{m_0+1}_1|M^{m_0+1}_1)\right]\nonumber\\
&\qquad +\frac{1}{n-m_0-1}B_{m_0+1}+\frac{1}{n-m_0-1}H(S_{m_0+2\rightarrow1},\ldots,S_{n\rightarrow1},W^{m_0+1}_1|M^{m_0+1}_1)\nonumber\\
&= \sum_{k=1}^{m_0}\frac{(n-2)(n-k)+1}{k(n-k)}{B}_k+\left[\frac{n-2}{m_0+1}-\frac{m_0}{(n-m_0-1)(m_0+1)}+\frac{1}{n-m_0-1}\right]B_{m_0+1}\nonumber\\
&\qquad+\left[\frac{n-2}{m_0+1}-\frac{m_0}{(n-m_0-1)(m_0+1)}\right]H(W^{m_0+1}_1|M^{m_0+1}_1)\nonumber\\
&\qquad\qquad+\frac{1}{n-m_0-1}H(S_{m_0+2\rightarrow1},\ldots,S_{n\rightarrow1},W^{m_0+1}_1|M^{m_0+1}_1)
\end{align}
which is precisely (\ref{eqn:claim}) for $m=m_0+1$ by noting that
\begin{align}
\left[\frac{n-2}{m_0+1}-\frac{m_0}{(n-m_0-1)(m_0+1)}+\frac{1}{n-m_0-1}\right]=\frac{(n-2)(n-m_0-1)+1}{(n-m_0-1)(m_0+1)}.
\end{align}
Here, $(a)$ follows from \eqref{eq:msp}, and $(b)$ follows, once again, from the conditional version of Han's inequality \cite{Han-IC78} under the symmetry assumption:
\begin{align}
\frac{1}{m_0}H(W^{m_0}_1|M^{m_0}_1) \geq \frac{1}{m_0+1}H(W^{m_0+1}_1|M^{m_0}_1)
\end{align}
and the fact that 
\begin{align*}
\left[\frac{n-2}{m_0}-\frac{1}{n-m_0-1}\right] \geq 0, \quad \forall m_0=1,2,\ldots,d-1
\end{align*}
and $(c)$ follows from the simple fact that
\begin{align*}
H(W^{m_0+1}_1|M^{m_0}_1)&=H(W^{m_0+1}_1,M_{m_0+1}|M^{m_0}_1)\\
&=H(M_{m_0+1}|M^{m_0}_1)+H(W^{m_0+1}_1|M^{m_0+1}_1)\\
&=B_{m_0+1}+H(W^{m_0+1}_1|M^{m_0+1}_1).
\end{align*}
This completes the proof of Theorem \ref{theorem:msp}.
\end{proof}

Readers familiar with the converse proof for the (symmetrical) MLD coding problem may recognize certain similarities between the above proof and that in \cite{RocheYeungHau:97}, particularly the use of Han's inequality. The key difference is that in the converse proof in \cite{RocheYeungHau:97}, one only needs to \lq\lq{}peel\rq\rq{} off the message rates by combining information in $W_i$'s sequentially. Here, however, the regeneration requirement necessitates a much more elaborate peeling process.

\section{The Tradeoff Rate Region $\mathcal{R}_4(\bar{B}_1,\bar{B}_2,\bar{B}_3)$: Proof of Theorem \ref{theorem:MRR4}}	
					
In this section the proof of Theorem \ref{theorem:MRR4} is presented. We start by providing a new code construction that achieves a particular normalized rate point, which will play a crucial role in the proof of Theorem \ref{theorem:MRR4}. 
					
\label{sec:case4}
\subsection{A New Code}
We first prove the following proposition. 

\begin{prop}
\label{theorem:case4outside}
The normalized rate pair $\left(\frac{4}{9},\frac{2}{9}\right)\in\mathcal{R}_4\left(0,\frac{1}{3},\frac{2}{3}\right)$.
\end{prop}

\begin{proof}
We give a novel code construction where $B_1=0$, $B_2=3$, $B_3=6$, $\alpha=4$, and $\beta=2$. For concreteness, the code symbols and algebraic operations are assumed in GF$(2^4)$. Let us denote the information symbols of message $M_2$ as $(x_1,x_2,x_3)$, and the symbols of message $M_3$ as $(y_1,y_2,...,y_6)$. 

\vspace{0.1cm}
\textbf{Encoding: } First use a $(10,3)$ MDS erasure code ({\em e.g.}, Reed-Solomon code) to encode $(x_1,x_2,x_3)$ into ten coded symbols $(z_1,z_2,...,z_{10})$, such that any three symbols can completely recover $(x_1,x_2,x_3)$. Then place  linear combinations of $(y_1,y_2,...,y_6)$ and $(z_1,z_2,...,z_{10})$ into the nodes as in Table \ref{tab:codesymbols}, where the addition $+$ is also in GF$(2^4)$.

\vspace{0.1cm}

\textbf{Decoding $M_2$ using any two nodes: } To decode $M_2$, observe that any pair of nodes has two symbols involving the same $y_j$, in the form of $z_i+y_j$ in one node and $y_j$ in the other node. For example node 2 has symbol $z_6+y_4$ and node 3 has $y_4$. This implies $z_i$ can be recovered, and together with the first symbols stored in this pair of nodes,  we have three distinct symbols in the set $\{z_1,z_2,...,z_{10}\}$. Thus by the property of the MDS code, these three symbols can be used to recover $(x_1,x_2,x_3)$ and thus the message $M_2$.

\vspace{0.1cm}

\textbf{Decoding $M_2$ and $M_3$ using any three nodes: } Recall using any two nodes we can recover the message $M_2$, and thus all the code symbols $(z_1,z_2,...,z_{10})$. This implies that when three nodes are available, we can eliminate all the $z_i$ symbols first in the linear combinations. However, it is clear that after this step all the symbols $(y_1,y_2,...,y_6)$ are directly available, and thus the message $M_3$ can be decoded.

\vspace{0.1cm}
\textbf{Repair using any three nodes: } To regenerate the symbols in one node from the other three, each of the helper nodes sends the first symbol stored on the nodes as the initial step. Denote the $y$ symbols on the failed node as $(y_i,y_j,y_k)$, which may be stored in a form also involving $z$-symbols. The helper nodes each find in the symbols stored on it the one involving $y_i$, $y_j$ and $y_k$, respectively, and send these symbol combinations as the second step. The placement of the $y$-symbol guarantees that these symbols are stored on the three helper nodes respectively. Recall that from any three $z$-symbols available in the initial step, the message $M_2$ can recovered, and thus any of the $z$-symbols. This implies that $(y_i,y_j,y_k)$ can be recovered after eliminating the $z$ symbols from the received symbol combinations in the second step, and thus all the symbols on the failed node can be successfully regenerated. Each helper node contributes exactly $2$ symbols in this process. 
\end{proof}

\begin{table}
\centering
\caption{A code for $n=4$ where $(\alpha,\beta)=(4,2)$ and $(B_1,B_2,B_3)=(0,3,6)$.}
\label{tab:codesymbols}
\begin{tabular}{|c|c|c|c|c|}
\hline
         &symbol 1&symbol 2&symbol 3&symbol 4\\\hline
node 1&$z_1$&$z_5+y_1$&$y_2$&$y_3$\\\hline
node 2&$z_2$&$z_6+y_4$&$y_1$&$y_5$\\\hline
node 3&$z_3$&$z_7+y_2$&$y_6+z_{10}$&$y_4$\\\hline
node 4&$z_4$&$z_8+y_3$&$y_5+z_9$&$y_6$\\\hline
\end{tabular}
\end{table}

In the above code the linear combinations of $z$-symbols and $y$-symbols are not necessarily in GF$(2^4)$, and it can even be in the base field GF$(2)$. The only constraint on the alphabet size is through requiring the existence of an appropriate MDS erasure codes when encoding $(x_1,x_2,x_3)$, and we have chosen GF$(2^4)$ for simplicity. Readers familiar with the work \cite{RashmiShah:12:1} may recognize that if message $M_2=(x_1,x_2,x_3)$ does not exist, the remaining code involving $y$\rq{}s is a degenerate case of the repair-by-transfer code in \cite{RashmiShah:12:1}. 

As we shall see next, mixing can improve the storage-repair-bandwidth tradeoff over separate coding. However, such an improvement may come at the expense of requiring downloading more symbols than the separate-coding solution, when only a single message is required. Interestingly, in the above code, this potential drawback can be eliminated by downloading only the \lq\lq{}non-mixing\rq\rq{} symbols when all nodes are functioning.

\subsection{Forward Proof of Theorem \ref{theorem:MRR4}}

To prove the forward part of Theorem \ref{theorem:MRR4}, we examine the extreme points of the region $\mathcal{R}_4(\bar{B}_1,\bar{B}_2,\bar{B}_3)$ for fixed $(\bar{B}_1,\bar{B}_2,\bar{B}_3)$. It should be noted that for different values of $(\bar{B}_1,\bar{B}_2,\bar{B}_3)$, the extreme points may be different, which causes certain complications for the discussion. The intersecting points of any two inequalities in (\ref{eqn:rateregion4_1})-(\ref{eqn:rateregion4_6}), when taken to be equality, are listed in Table \ref{tab:MRR4} for any fixed $(\bar{B}_1,\bar{B}_2,\bar{B}_3)$. The first two columns specify which two inequalities induce the intersection.
\begin{normalsize}
\setlength{\extrarowheight}{3pt}
\begin{table}
\caption{Analysis of extreme points as intersections of any two inequalities.\label{tab:MRR4}}
\centering
\begin{tabular}{|c|c|c|c|c|}
\hline
\multicolumn{2}{|c|}{Intersection of}&$\bar{\alpha}$&$\bar{\beta}$&remark\\\hline
\multirow{5}{*}{(\ref{eqn:rateregion4_1})}&(\ref{eqn:rateregion4_2})&\multirow{5}{*}{$\bar{B}_1+\frac{\bar{B}_2}{2}+\frac{\bar{B}_3}{3}$}&$\frac{\bar{B}_1}{3}+\frac{\bar{B}_2}{4}+\frac{\bar{B}_3}{3}$&$(a)$\\\cline{2-2}\cline{4-5}
&(\ref{eqn:rateregion4_3})&&$\frac{\bar{B}_1}{3}+\frac{\bar{B}_2}{4}+\frac{7\bar{B}_3}{24}$&$(b)$\\\cline{2-2}\cline{4-4}
&(\ref{eqn:rateregion4_4})&&$\frac{\bar{B}_1}{3}+\frac{2\bar{B}_2}{9}+\frac{5\bar{B}_3}{18}$&(c)\\\cline{2-2}\cline{4-4}
&(\ref{eqn:rateregion4_5})&&$\frac{\bar{B}_1}{3}+\frac{\bar{B}_2}{4}+\frac{\bar{B}_3}{4}$&(d)\\\cline{2-2}\cline{4-4}
&(\ref{eqn:rateregion4_6})&&$\frac{\bar{B}_1}{3}+\frac{\bar{B}_2}{5}+\frac{\bar{B}_3}{6}$&(e)\\\hline
\multirow{4}{*}{(\ref{eqn:rateregion4_2})}&(\ref{eqn:rateregion4_3})&{$\bar{B}_1+\frac{\bar{B}_2}{2}+\frac{3\bar{B}_3}{8}$}&$\frac{\bar{B}_1}{3}+\frac{\bar{B}_2}{4}+\frac{\bar{B}_3}{4}$&$(f)$\\\cline{2-5}
&(\ref{eqn:rateregion4_4})&$\bar{B}_1+\frac{25\bar{B}_2}{48}+\frac{3\bar{B}_3}{8}$&$\frac{\bar{B}_1}{3}+\frac{5\bar{B}_2}{24}+\frac{\bar{B}_3}{4}$&$(g)$\\\cline{2-4}
&(\ref{eqn:rateregion4_5})&$\bar{B}_1+\frac{\bar{B}_2}{2}+\frac{7\bar{B}_3}{18}$&$\frac{\bar{B}_1}{3}+\frac{\bar{B}_2}{4}+\frac{2\bar{B}_3}{9}$&$(h)$\\\cline{2-4}
&(\ref{eqn:rateregion4_6})&$\bar{B}_1+\frac{21\bar{B}_2}{40}+\frac{5\bar{B}_3}{12}$&$\frac{\bar{B}_1}{3}+\frac{\bar{B}_2}{5}+\frac{\bar{B}_3}{6}$&$(i)$\\\hline
\multirow{3}{*}{(\ref{eqn:rateregion4_3})}&(\ref{eqn:rateregion4_4})&{$\bar{B}_1+\frac{7\bar{B}_2}{12}+\frac{3\bar{B}_3}{8}$}&$\frac{\bar{B}_1}{3}+\frac{\bar{B}_2}{6}+\frac{\bar{B}_3}{4}$&$(j)$\\\cline{2-5}
&(\ref{eqn:rateregion4_5})&$\bar{B}_1+\frac{\bar{B}_2}{2}+\frac{5\bar{B}_3}{12}$&$\frac{\bar{B}_1}{3}+\frac{\bar{B}_2}{4}+\frac{5\bar{B}_3}{24}$&$(k)$\\\cline{2-5}
&(\ref{eqn:rateregion4_6})&$\bar{B}_1+\frac{11\bar{B}_2}{20}+\frac{11\bar{B}_3}{24}$&$\frac{\bar{B}_1}{3}+\frac{\bar{B}_2}{5}+\frac{\bar{B}_3}{6}$&$(l)$\\\hline
\multirow{2}{*}{(\ref{eqn:rateregion4_4})}&(\ref{eqn:rateregion4_5})&$\bar{B}_1+\frac{\bar{B}_2}{3}+\frac{\bar{B}_3}{2}$&$\frac{\bar{B}_1}{3}+\frac{\bar{B}_2}{3}+\frac{\bar{B}_3}{6}$&$(m)$\\\cline{2-5}
&(\ref{eqn:rateregion4_6})&$\bar{B}_1+\frac{8\bar{B}_2}{15}+\frac{\bar{B}_3}{2}$&$\frac{\bar{B}_1}{3}+\frac{\bar{B}_2}{5}+\frac{\bar{B}_3}{6}$&$(n)$\\\hline
(\ref{eqn:rateregion4_5})&(\ref{eqn:rateregion4_6})&\multirow{1}{*}{$\bar{B}_1+\frac{3\bar{B}_2}{5}+\frac{\bar{B}_3}{2}$}&$\frac{\bar{B}_1}{3}+\frac{\bar{B}_2}{5}+\frac{\bar{B}_3}{6}$&$(o)$\\\hline
\end{tabular}
\end{table}
\end{normalsize}

\begin{itemize}
\item Point $(a)$ can be achieved by a separate-coding scheme that uses $(4,1,3)$, $(4,2,3)$ and $(4,3,3)$ regenerating codes operating at the normalized rate pairs $(\bar{\alpha}_1,\bar{\beta}_1)=\left(1,\frac{1}{3}\right)$, $(\bar{\alpha}_2,\bar{\beta}_2)=\left(\frac{1}{2},\frac{1}{4}\right)$ and $(\bar{\alpha}_3,\bar{\beta}_3)=\left(\frac{1}{3},\frac{1}{3}\right)$, respectively.
\item By the inequality (\ref{eqn:rateregion4_2}), points $(b)$ and $(d)$ are feasible only if $\bar{B}_3=0$, and points $(c)$ and $(e)$ are feasible only if $\bar{B}_2=\bar{B}_3=0$. In both cases, these points are reduced to point $(a)$, which has been shown to be achievable.
\item Point $(f)$ can be achieved by a separate-coding scheme that uses $(4,1,3)$, $(4,2,3)$ and $(4,3,3)$ regenerating codes operating at at the normalized rate pairs $(\bar{\alpha}_1,\bar{\beta}_1)=\left(1,\frac{1}{3}\right)$, $(\bar{\alpha}_2,\bar{\beta}_2)=\left(\frac{1}{2},\frac{1}{4}\right)$ and $(\bar{\alpha}_3,\bar{\beta}_3)=\left(\frac{3}{8},\frac{1}{4}\right)$, respectively.
\item By the inequality (\ref{eqn:rateregion4_3}), point $(g)$ is feasible only if $\bar{B}_2=0$; point $(h)$ is feasible only if $\bar{B}_3=0$; and point $(i)$ is feasible only if $\bar{B}_2=\bar{B}_3=0$. In all three cases, these points are reduced to point $(f)$, which has been shown to be achievable.
\item By the inequality (\ref{eqn:rateregion4_5}), point $(j)$ is feasible only if $\bar{B}_2\leq \frac{\bar{B}_3}{2}$. Consider a message triple $(M_1,M_2,M_3)$ with sufficiently large message rates $\left(B_1,B_2,B_3\right)$, where
\begin{align*}
B_2=\frac{1-\gamma}{2}B_3, \quad \gamma \in [0,1].
\end{align*}
Split message $M_3$ into two independent sub-messages $M_{3,1}$ and $M_{3,2}$ with rates $\gamma B_3$ and $(1-\gamma)B_3=2B_2$, respectively. Consider encoding the messages $M_1$, $M_{3,1}$ and $(M_2,M_{3,2})$ separately. More specifically, encode message $M_1$ using a $(4,1,3)$ regenerating code operating at the normalized rate pair $\left(1,\frac{1}{3}\right)$; encode message $M_{3,1}$ using a $(4,3,3)$ regenerating code operating at the normalized rate pair $\left(\frac{3}{8},\frac{1}{4}\right)$; and encode the messages $(M_2,M_{3,2})$ jointly using a code as described in Proposition~\ref{theorem:case4outside} operating at the normalized rate pair $\left(\frac{4}{9},\frac{2}{9}\right)$. The total rates of this coding scheme are given by:
\begin{align*}
(\alpha,\beta) &=\left(B_1+\frac{3\gamma B_3}{8}+\frac{4(B_2+2B_2)}{9},\frac{B_1}{3}+\frac{\gamma B_3}{4}+\frac{2(B_2+2B_2)}{9}\right)\\
&= \left(B_1+\frac{7B_2}{12}+\frac{3B_3}{8},\frac{B_1}{3}+\frac{B_2}{6}+\frac{B_3}{4}\right).
\end{align*}
Normalizing both sides by $B_1+B_2+B_3$, we conclude that any normalized rate pair $(j)$ with
\begin{align*}
\bar{B}_2=\frac{1-\gamma}{2}\bar{B}_3, \quad \gamma \in [0,1]
\end{align*}
is achievable, implying that point $(j)$  is indeed achievable whenever $\bar{B}_2\leq \frac{\bar{B}_3}{2}$.
\item By the inequality (\ref{eqn:rateregion4_4}), point $(k)$ is feasible only if $\bar{B}_3 \leq 2\bar{B}_2$. Consider a message triple $(M_1,M_2,M_3)$ with sufficiently large message rates $\left(B_1,B_2,B_3\right)$, where
\begin{align*}
B_3=2(1-\gamma)B_2, \quad \gamma \in [0,1].
\end{align*}
Split message $M_2$ into two independent sub-messages $M_{2,1}$ and $M_{2,2}$ with rates $\gamma B_2$ and $(1-\gamma)B_2=\frac{B_3}{2}$, respectively. Consider encoding the messages $M_1$, $M_{2,1}$ and $(M_{2,2},M_3)$ separately. More specifically, encode message $M_1$ using a $(4,1,3)$ regenerating code operating at the normalized rate pair $\left(1,\frac{1}{3}\right)$; encode message $M_{2,1}$ using a $(4,2,3)$ regenerating code operating at the normalized rate pair $\left(\frac{1}{2},\frac{1}{4}\right)$; and encode the messages $(M_{2,2},M_3)$ jointly using a code as described in Proposition~\ref{theorem:case4outside} operating at the normalized rate pair $\left(\frac{4}{9},\frac{2}{9}\right)$. The total rates of this coding scheme are given by:
\begin{align*}
(\alpha,\beta) &=\left(B_1+\frac{\gamma B_2}{2}+\frac{4\left(\frac{B_3}{2}+B_3\right)}{9},\frac{B_1}{3}+\frac{\gamma B_2}{4}+\frac{2\left(\frac{B_3}{2}+B_3\right)}{9}\right)\\
&= \left(B_1+\frac{B_2}{2}+\frac{5B_3}{12},\frac{B_1}{3}+\frac{B_2}{4}+\frac{5B_3}{24}\right).
\end{align*}
Normalizing both sides by $B_1+B_2+B_3$, we may conclude that any normalized rate pair $(k)$ with
\begin{align*}
\bar{B}_3=2(1-\gamma)\bar{B}_2, \quad \gamma \in [0,1]
\end{align*}
is achievable, implying point $(k)$ is indeed achievable whenever $\bar{B}_3\leq 2\bar{B}_2$.
\item By the inequalities (\ref{eqn:rateregion4_4}) and (\ref{eqn:rateregion4_5}), point $(l)$ is feasible only if $\bar{B}_2=\bar{B}_3=0$.  In this case, point $(l)$ is reduced to point $(a)$, which has been shown to be achievable.
\item By the inequality (\ref{eqn:rateregion4_4}), point $(m)$ is feasible only if $\bar{B}_2\leq \frac{\bar{B}_3}{2}$. Consider a message triple $(M_1,M_2,M_3)$ with sufficiently large message rates $\left(B_1,B_2,B_3\right)$, where
\begin{align*}
B_2=\frac{1-\gamma}{2}B_3, \quad \gamma \in [0,1].
\end{align*}
Split message $M_3$ into two independent sub-messages $M_{3,1}$ and $M_{3,2}$ with rates $\gamma B_3$ and $(1-\gamma)B_3=2B_2$, respectively. Consider encoding the messages $M_1$, $M_{3,1}$ and $(M_2,M_{3,2})$ separately. More specifically, encode message $M_1$ using a $(4,1,3)$ regenerating code operating at the normalized rate pair $\left(1,\frac{1}{3}\right)$; encode message $M_{3,1}$ using a $(4,3,3)$ regenerating code operating at the normalized rate pair $\left(\frac{1}{2},\frac{1}{6}\right)$; and encode the messages $(M_2,M_{3,2})$ jointly using a code as described in Proposition~\ref{theorem:case4outside} operating at the normalized rate pair $\left(\frac{4}{9},\frac{2}{9}\right)$. The total rates of this coding scheme are given by:
\begin{align*}
(\alpha,\beta) &=\left(B_1+\frac{\gamma B_3}{2}+\frac{4(B_2+2B_2)}{9},\frac{B_1}{3}+\frac{\gamma B_3}{6}+\frac{2(B_2+2B_2)}{9}\right)\\
&= \left(B_1+\frac{B_2}{3}+\frac{B_3}{2},\frac{B_1}{3}+\frac{B_2}{3}+\frac{B_3}{6}\right).
\end{align*}
Normalizing both sides by $B_1+B_2+B_3$, we may conclude that any normalized rate pair $(m)$ with
\begin{align*}
\bar{B}_2=\frac{1-\gamma}{2}\bar{B}_3, \quad \gamma \in [0,1]
\end{align*}
is achievable, implying point $(m)$ is indeed achievable whenever $\bar{B}_2\leq \frac{\bar{B}_3}{2}$.
\item By the inequality (\ref{eqn:rateregion4_5}), point $(n)$ is feasible only if $\bar{B}_2=0$. In this case, point $(n)$ can be achieved by a separate-coding scheme that uses $(4,1,3)$ and $(4,3,3)$ regenerating codes operating at the normalized rate pairs $(\bar{\alpha}_1,\bar{\beta}_1)=\left(1,\frac{1}{3}\right)$ and $(\bar{\alpha}_3,\bar{\beta}_3)=\left(\frac{1}{2},\frac{1}{6}\right)$, respectively.
\item Finally, point $(o)$ can be achieved by a separate-coding scheme that uses $(4,1,3)$, $(4,2,3)$ and $(4,3,3)$ regenerating codes operating at the normalized rate pairs $(\bar{\alpha}_1,\bar{\beta}_1)=\left(1,\frac{1}{3}\right)$, $(\bar{\alpha}_2,\bar{\beta}_2)=\left(\frac{3}{5},\frac{1}{5}\right)$ and $(\bar{\alpha}_3,\bar{\beta}_3)=\left(\frac{1}{2},\frac{1}{6}\right)$, respectively.
\end{itemize}

The proof is now complete.\qed

\subsection{Converse Proof of Theorem \ref{theorem:MRR4}}
 
To establish the converse of Theorem \ref{theorem:MRR4}, we shall prove that every rate pair $(\bar{\alpha},\bar{\beta}) \in \mathcal{R}_4(\bar{B}_1,\bar{B}_2,\bar{B}_3)$ must satisfy the inequalities (\ref{eqn:rateregion4_1})--(\ref{eqn:rateregion4_6}).  The inequality (\ref{eqn:rateregion4_1}) holds even without the regeneration requirement\cite{RocheYeungHau:97}, and the inequality (\ref{eqn:rateregion4_2}) follows directly from Theorem \ref{theorem:msp} by setting $n=4$. It remains to show that the inequalities (\ref{eqn:rateregion4_3})--(\ref{eqn:rateregion4_6}) are true, and we shall prove each as a separate proposition. Instead of writing the proofs in the conventional fashion as chains of inequalities, we utilize the computational approach developed in \cite{Tian:JSAC13} and prove these inequalities by tabulation. Our proof of each inequality is given as two tables: The first one lists the joint entropy terms in the proof, and the second one lists the coefficients of needed inequalities. The last row, as the summation of all the other rows, is exactly the sought-after inequality. Note that each row, except for the last one, in the second table is a ``simple" Shannon-type inequality, possibly after a permutation of the indices for each entropy term. For example, the third line of Table \ref{table:MRR4_32} is
\begin{align}
2H(S_{4\rightarrow 3})+2H(S_{4\rightarrow 2},S_{3\rightarrow 2})-2H(S_{4\rightarrow 1},S_{3\rightarrow 1},S_{2\rightarrow 1})\geq 0
\end{align}
which is equivalent to the simple independence bound on entropy:
\begin{align}
H(S_{2\rightarrow 1})+H(S_{4\rightarrow 1},S_{3\rightarrow 1})-H(S_{4\rightarrow 1},S_{3\rightarrow 1},S_{2\rightarrow 1})\geq 0
\end{align}
after taking into account of the symmetry $H(S_{4\rightarrow 3})=H(S_{2\rightarrow 1})$ and $H(S_{4\rightarrow 2},S_{3\rightarrow 2})=H(S_{4\rightarrow 1},S_{3\rightarrow 1})$ in the assumed  solution set. Further note that for $n=4$ we have $\bar{B}^+_3=1$, however we still write it as $\bar{B}^+_3$ to make explicit its meaning of sum rate.

\begin{prop}
For any $(\bar{\alpha},\bar{\beta})\in \mathcal{R}_4(\bar{B}_1,\bar{B}_2,\bar{B}_3)$, we have \label{prop:MRR4_3}
\begin{align}
24(\bar{\alpha}+\bar{\beta})&\geq 14\bar{B}^+_1+3\bar{B}^+_2+15\bar{B}^+_3=24\left(\frac{4}{3}\bar{B}_1+\frac{3}{4}\bar{B}_2+\frac{5}{8}\bar{B}_3\right).
\end{align}
\end{prop}
\begin{proof}
See Tables \ref{table:MRR4_31} and \ref{table:MRR4_32}.
\end{proof}
\begin{table}
\centering
\caption{Terms needed to prove Proposition \ref{prop:MRR4_3}\label{table:MRR4_31}.}
\begin{tabular}{|c|c|}
\hline
$T_{ 1}$ & $H(S_{4\rightarrow3})$ \\
$T_{ 2}$ & $H(S_{4\rightarrow2},S_{3\rightarrow2})$ \\
$T_{ 3}$ & $H(S_{4\rightarrow1},S_{3\rightarrow1},S_{2\rightarrow1})$ \\
$T_{ 4}$ & $H(W_4)$ \\
$T_{ 5}$ & $H(S_{3\rightarrow2},W_4)$ \\
$T_{ 6}$ & $H(S_{3\rightarrow4},S_{2\rightarrow4},W_4)$ \\
$T_{ 7}$ & $H(W_4,W_3)$ \\
$T_{ 8}$ & $H(S_{2\rightarrow4},W_4,W_3)$ \\
$T_{ 9}$ & $H(S_{3\rightarrow1},S_{2\rightarrow1},W_4,W_1)$ \\
$T_{10}$ & $H(S_{4\rightarrow3},M_1)$ \\
$T_{11}$ & $H(S_{4\rightarrow2},S_{3\rightarrow2},M_1)$ \\
$T_{12}$ & $H(W_4,M_1,M_2)$ \\
$T_{13}$ & $H(S_{3\rightarrow2},W_4,M_1,M_2)$ \\
$T_{14}$ & $H(S_{3\rightarrow4},W_4,M_1,M_2)$ \\
$T_{15}$ & $H(S_{3\rightarrow2},S_{2\rightarrow4},W_4,M_1,M_2)$ \\
$T_{16}$ & $H(M_1)=B^+_1$ \\
$T_{17}$ & $H(M_1,M_2)=B^+_2$ \\
$T_{18}$ & $H(M_1,M_2,M_3)=B^+_3$ \\
\hline
\end{tabular}
\end{table}

\begin{table}
\centering
\caption{Proof by Tabulation of Proposition \ref{prop:MRR4_3}, with terms defined in Table \ref{table:MRR4_31}.\label{table:MRR4_32}}
 \centering 
        \setlength{\tabcolsep}{3pt}
        \begin{tabular}{|cccccccccccccccccc|}
\hline
$T_{ 1}$  &$T_{ 2}$  &$T_{ 3}$  &$T_{ 4}$  &$T_{ 5}$  &$T_{ 6}$  &$T_{ 7}$  &$T_{ 8}$  &$T_{ 9}$  &$T_{10}$  &$T_{11}$  &$T_{12}$  &$T_{13}$  &$T_{14}$  &$T_{15}$  &$T_{16}$  &$T_{17}$  &$T_{18}$  \\\hline
$ 16$     &$ -8$     &          &          &          &          &          &          &          &          &          &          &          &          &          &          &          &           \\
          &          &          &$  4$     &$ -4$     &          &          &          &          &$  4$     &          &          &          &          &          &$ -4$     &          &           \\
$  2$     &$  2$     &$ -2$     &          &          &          &          &          &          &          &          &          &          &          &          &          &          &           \\
          &          &          &$ 12$     &          &          &$ -6$     &          &          &          &          &          &          &          &          &$ -6$     &          &           \\
          &          &          &          &$  2$     &$  2$     &          &$ -2$     &          &          &$ -2$     &          &          &          &          &          &          &           \\
$  6$     &$  6$     &          &          &          &$ -6$     &          &          &          &          &          &          &          &          &          &          &          &           \\
          &          &          &$  4$     &          &$  4$     &          &$ -4$     &          &$ -4$     &          &          &          &          &          &          &          &           \\
          &          &          &$  4$     &          &          &          &          &$ -4$     &          &$  4$     &          &          &          &          &$ -4$     &          &           \\
          &          &          &          &          &          &          &          &          &          &          &$  3$     &          &$  3$     &          &          &$ -3$     &$ -3$     \\
          &          &          &          &          &          &$  3$     &          &          &          &          &$ -3$     &$  3$     &          &          &          &          &$ -3$     \\
          &          &          &          &          &          &$  3$     &          &$  3$     &          &          &          &$ -3$     &          &          &          &          &$ -3$     \\
          &          &          &          &          &          &          &$  3$     &          &          &          &          &          &$ -3$     &$  3$     &          &          &$ -3$     \\
          &          &          &          &          &          &          &$  3$     &$  3$     &          &          &          &          &          &$ -3$     &          &          &$ -3$     \\
          &          &$  2$     &          &$  2$     &          &          &          &$ -2$     &          &$ -2$     &          &          &          &          &          &          &           \\
\hline
\hline
$24$&&&$24$&&&&&&&&&&&&$-14$&$-3$&$-15$\\
\hline
\end{tabular}
\end{table}

\begin{prop}
For any $(\bar{\alpha},\bar{\beta})\in \mathcal{R}_4(\bar{B}_1,\bar{B}_2,\bar{B}_3)$, we have \label{prop:MRR4_4}
\begin{align}
6(2\bar{\alpha}+3\bar{\beta})&\geq 8\bar{B}^+_1+\bar{B}^+_2+9\bar{B}^+_3= 6\left(3\bar{B}_1+\frac{5}{3}\bar{B}_2+\frac{3}{2}\bar{B}_3\right).
\end{align}
\end{prop}
\begin{proof}
See Tables \ref{table:MRR4_41} and \ref{table:MRR4_42}.
\end{proof}
\begin{table}
\centering
\caption{Terms needed to prove Proposition \ref{prop:MRR4_4}\label{table:MRR4_41}.}
\begin{tabular}{|c|c|}
\hline
$T_{ 1}$ & $H(S_{4\rightarrow3})$ \\
$T_{ 2}$ & $H(S_{4\rightarrow2},S_{3\rightarrow2})$ \\
$T_{ 3}$ & $H(S_{4\rightarrow1},S_{3\rightarrow1},S_{2\rightarrow1})$ \\
$T_{ 4}$ & $H(W_4)$ \\
$T_{ 5}$ & $H(S_{3\rightarrow2},W_4)$ \\
$T_{ 6}$ & $H(S_{3\rightarrow4},W_4)$ \\
$T_{ 7}$ & $H(S_{3\rightarrow4},S_{2\rightarrow4},W_4)$ \\
$T_{ 8}$ & $H(W_4,W_3)$ \\
$T_{ 9}$ & $H(S_{2\rightarrow4},W_4,W_3)$ \\
$T_{10}$ & $H(S_{3\rightarrow1},S_{2\rightarrow1},W_4,W_1)$ \\
$T_{11}$ & $H(S_{4\rightarrow3},M_1)$ \\
$T_{12}$ & $H(S_{4\rightarrow2},S_{3\rightarrow2},M_1)$ \\
$T_{13}$ & $H(S_{4\rightarrow3},S_{3\rightarrow4},M_1)$ \\
$T_{14}$ & $H(S_{4\rightarrow3},S_{3\rightarrow4},S_{2\rightarrow4},M_1)$ \\
$T_{15}$ & $H(W_4,M_1,M_2)$ \\
$T_{16}$ & $H(S_{3\rightarrow2},W_4,M_1,M_2)$ \\
$T_{17}$ & $H(S_{3\rightarrow4},W_4,M_1,M_2)$ \\
$T_{18}$ & $H(S_{3\rightarrow2},S_{2\rightarrow4},W_4,M_1,M_2)$ \\
$T_{19}$ & $H(M_1)=B^+_1$ \\
$T_{20}$ & $H(M_1,M_2)=B^+_2$ \\
$T_{21}$ & $H(M_1,M_2,M_3)=B^+_3$ \\
\hline
\end{tabular}
\end{table}

\begin{table}
\centering
\caption{Proof by Tabulation of Proposition \ref{prop:MRR4_4}, with terms defined in Table \ref{table:MRR4_41}.\label{table:MRR4_42}}
 \centering 
        \setlength{\tabcolsep}{3pt}
        \begin{tabular}{|ccccccccccccccccccccc|}
\hline
$T_{ 1}$  &$T_{ 2}$  &$T_{ 3}$  &$T_{ 4}$  &$T_{ 5}$  &$T_{ 6}$  &$T_{ 7}$  &$T_{ 8}$  &$T_{ 9}$  &$T_{10}$  &$T_{11}$  &$T_{12}$  &$T_{13}$  &$T_{14}$  &$T_{15}$  &$T_{16}$  &$T_{17}$  &$T_{18}$  &$T_{19}$  &$T_{20}$  &$T_{21}$  \\\hline
$ 12$     &$ -6$     &          &          &          &          &          &          &          &          &          &          &          &          &          &          &          &          &          &          &           \\
          &          &          &$  4$     &$ -4$     &          &          &          &          &          &$  4$     &          &          &          &          &          &          &          &$ -4$     &          &           \\
$  4$     &$  4$     &$ -4$     &          &          &          &          &          &          &          &          &          &          &          &          &          &          &          &          &          &           \\
          &          &          &$  4$     &          &          &          &          &          &$ -4$     &          &$  4$     &          &          &          &          &          &          &$ -4$     &          &           \\
          &          &          &$  2$     &          &$  2$     &          &$ -2$     &          &          &$ -2$     &          &          &          &          &          &          &          &          &          &           \\
          &          &          &          &$  2$     &          &$  2$     &          &$ -2$     &          &          &$ -2$     &          &          &          &          &          &          &          &          &           \\
$  2$     &$  2$     &          &          &          &          &$ -2$     &          &          &          &          &          &          &          &          &          &          &          &          &          &           \\
          &          &          &$  2$     &          &$ -2$     &          &          &          &          &$ -2$     &          &$  2$     &          &          &          &          &          &          &          &           \\
          &          &$  2$     &          &$  2$     &          &          &          &          &$ -2$     &          &$ -2$     &          &          &          &          &          &          &          &          &           \\
          &          &          &          &          &          &          &          &          &$  2$     &          &          &$ -2$     &$  2$     &          &          &          &          &          &          &$ -2$     \\
          &          &$  2$     &          &          &          &          &          &          &$  2$     &          &          &          &$ -2$     &          &          &          &          &          &          &$ -2$     \\
          &          &          &          &          &          &          &          &          &          &          &          &          &          &$  1$     &          &$  1$     &          &          &$ -1$     &$ -1$     \\
          &          &          &          &          &          &          &$  1$     &          &          &          &          &          &          &$ -1$     &$  1$     &          &          &          &          &$ -1$     \\
          &          &          &          &          &          &          &$  1$     &          &$  1$     &          &          &          &          &          &$ -1$     &          &          &          &          &$ -1$     \\
          &          &          &          &          &          &          &          &$  1$     &          &          &          &          &          &          &          &$ -1$     &$  1$     &          &          &$ -1$     \\
          &          &          &          &          &          &          &          &$  1$     &$  1$     &          &          &          &          &          &          &          &$ -1$     &          &          &$ -1$     \\
\hline
\hline
$18$&&&$12$&&&&&&&&&&&&&&&$-8$&$-1$&$-9$\\
\hline
\end{tabular}
\end{table}

\begin{prop}
For any $(\bar{\alpha},\bar{\beta})\in \mathcal{R}_4(\bar{B}_1,\bar{B}_2,\bar{B}_3)$, we have \label{prop:MRR4_5}
\begin{align}
12(\bar{\alpha}+2\bar{\beta})\geq 8\bar{B}^+_1+2\bar{B}^+_2+10\bar{B}^+_3=12\left(\frac{5}{3}\bar{B}_1+\bar{B}_2+\frac{5}{6}\bar{B}_3\right).
\end{align}
\end{prop}
\begin{proof}
See Tables \ref{table:MRR4_51} and \ref{table:MRR4_52}.
\end{proof}

\begin{table}
\centering
\caption{Terms needed to prove Proposition \ref{prop:MRR4_5}\label{table:MRR4_51}.}
\begin{tabular}{|c|c|}
\hline
$T_{ 1}$ & $H(S_{4\rightarrow3})$ \\
$T_{ 2}$ & $H(S_{4\rightarrow2},S_{3\rightarrow2})$ \\
$T_{ 3}$ & $H(S_{4\rightarrow1},S_{3\rightarrow1},S_{2\rightarrow1})$ \\
$T_{ 4}$ & $H(W_4)$ \\
$T_{ 5}$ & $H(S_{3\rightarrow4},W_4)$ \\
$T_{ 6}$ & $H(S_{3\rightarrow2},W_4)$ \\
$T_{ 7}$ & $H(S_{3\rightarrow4},S_{2\rightarrow4},W_4)$ \\
$T_{ 8}$ & $H(W_4,W_3)$ \\
$T_{ 9}$ & $H(S_{2\rightarrow4},W_4,W_3)$ \\
$T_{10}$ & $H(S_{3\rightarrow1},S_{2\rightarrow1},W_4,W_1)$ \\
$T_{11}$ & $H(S_{3\rightarrow4},S_{3\rightarrow1},S_{2\rightarrow1},W_4,W_1)$ \\
$T_{12}$ & $H(S_{4\rightarrow3},M_1)$ \\
$T_{13}$ & $H(S_{4\rightarrow3},M_1,M_2)$ \\
$T_{14}$ & $H(S_{4\rightarrow2},S_{3\rightarrow4},M_1,M_2)$ \\
$T_{15}$ & $H(S_{4\rightarrow2},S_{3\rightarrow2},M_1,M_2)$ \\
$T_{16}$ & $H(S_{3\rightarrow2},W_4,M_1,M_2)$ \\
$T_{17}$ & $H(S_{3\rightarrow4},W_4,M_2,M_2)$ \\
$T_{18}$ & $H(S_{3\rightarrow2},S_{2\rightarrow4},W_4,M_1,M_2)$ \\
$T_{19}$ & $H(S_{3\rightarrow2},S_{2\rightarrow3},W_4,M_1,M_2)$ \\
$T_{20}$ & $H(S_{3\rightarrow2},S_{2\rightarrow4},S_{1\rightarrow4},W_4,M_1,M_2)$ \\
$T_{21}$ & $H(M_1)=B^+_1$ \\
$T_{22}$ & $H(M_1,M_2)=B^+_2$ \\
$T_{23}$ & $H(M_1,M_2,M_3)=B^+_3$ \\
\hline
\end{tabular}
\end{table}

\begin{table}
\centering
\caption{Proof by Tabulation of Proposition \ref{prop:MRR4_5}, with terms defined in Table \ref{table:MRR4_51}.\label{table:MRR4_52}}
 \centering 
        \setlength{\tabcolsep}{3pt}
        \begin{tabular}{|ccccccccccccccccccccccc|}
\hline
$T_{ 1}$  &$T_{ 2}$  &$T_{ 3}$  &$T_{ 4}$  &$T_{ 5}$  &$T_{ 6}$  &$T_{ 7}$  &$T_{ 8}$  &$T_{ 9}$  &$T_{10}$  &$T_{11}$  &$T_{12}$  &$T_{13}$  &$T_{14}$  &$T_{15}$  &$T_{16}$  &$T_{17}$  &$T_{18}$  &$T_{19}$  &$T_{20}$  &$T_{21}$  &$T_{22}$  &$T_{23}$  \\
\hline
$ 16$     &$ -8$     &          &          &          &          &          &          &          &          &          &          &          &          &          &          &          &          &          &          &          &          &           \\
          &          &          &$  4$     &          &$ -4$     &          &          &          &          &          &$  4$     &          &          &          &          &          &          &          &          &$ -4$     &          &           \\
$  4$     &$  4$     &$ -4$     &          &          &          &          &          &          &          &          &          &          &          &          &          &          &          &          &          &          &          &           \\
          &          &          &$  1$     &$  1$     &          &          &$ -1$     &          &          &          &$ -1$     &          &          &          &          &          &          &          &          &          &          &           \\
$  4$     &$  4$     &          &          &          &          &$ -4$     &          &          &          &          &          &          &          &          &          &          &          &          &          &          &          &           \\
          &          &          &$  4$     &          &          &          &          &          &$ -4$     &          &          &          &          &$  4$     &          &          &          &          &          &$ -4$     &          &           \\
          &          &$  4$     &          &          &$  4$     &          &          &          &$ -4$     &          &          &          &          &$ -4$     &          &          &          &          &          &          &          &           \\
          &          &          &$  2$     &          &          &$  2$     &$ -2$     &          &          &          &$ -2$     &          &          &          &          &          &          &          &          &          &          &           \\
          &          &          &          &          &          &          &$  2$     &          &          &          &          &$  2$     &          &          &          &          &          &          &          &          &$ -2$     &$ -2$     \\
          &          &          &          &          &          &          &          &          &          &          &          &$ -1$     &          &          &$  1$     &$  1$     &          &          &          &          &          &$ -1$     \\
          &          &          &$  1$     &          &          &$  1$     &          &$ -1$     &          &          &$ -1$     &          &          &          &          &          &          &          &          &          &          &           \\
          &          &          &          &          &          &          &          &          &          &          &          &          &$ -1$     &          &          &          &          &$  1$     &$  1$     &          &          &$ -1$     \\
          &          &          &          &          &          &          &$  1$     &          &$  1$     &          &          &          &          &          &$ -1$     &          &          &          &          &          &          &$ -1$     \\
          &          &          &          &          &          &          &          &$  1$     &$  1$     &          &          &          &          &          &          &          &$ -1$     &          &          &          &          &$ -1$     \\
          &          &          &          &          &          &          &          &          &$  1$     &$  1$     &          &          &          &          &          &          &          &          &$ -1$     &          &          &$ -1$     \\
          &          &          &          &          &          &          &          &          &$  1$     &          &          &$ -1$     &$  1$     &          &          &          &          &          &          &          &          &$ -1$     \\
          &          &          &          &          &          &          &          &          &$  1$     &          &          &          &          &          &          &$ -1$     &$  1$     &          &          &          &          &$ -1$     \\
          &          &          &          &          &          &          &          &          &$  2$     &          &          &          &          &          &          &          &          &$ -1$     &          &          &          &$ -1$     \\
          &          &          &          &$ -1$     &          &$  1$     &          &          &$  1$     &$ -1$     &          &          &          &          &          &          &          &          &          &          &          &           \\
\hline\hline
$24$&&&$12$&&&&&&&&&&&&&&&&&$-8$&$-2$&$-10$\\
\hline
\end{tabular}
\end{table}

\begin{prop}
For any $(\bar{\alpha},\bar{\beta})\in \mathcal{R}_4(\bar{B}_1,\bar{B}_2,\bar{B}_3)$, we have \label{prop:MRR4_6}
\begin{align}
30\bar{\beta}\geq 4\bar{B}^+_1+\bar{B}^+_2+5\bar{B}^+_3=30\left(\frac{1}{3}\bar{B}_1+\frac{1}{5}\bar{B}_2+\frac{1}{6}\bar{B}_3\right).
\end{align}
\end{prop}
\begin{proof}
See Tables \ref{table:MRR4_61} and \ref{table:MRR4_62}.
\end{proof}
\begin{table}
\centering
\caption{Terms needed to prove Proposition \ref{prop:MRR4_6}\label{table:MRR4_61}.}
\begin{tabular}{|c|c|}
\hline
$T_{ 1}$ & $H(S_{4\rightarrow3})$ \\
$T_{ 2}$ & $H(S_{4\rightarrow2},S_{3\rightarrow2})$ \\
$T_{ 3}$ & $H(S_{4\rightarrow1},S_{3\rightarrow1},S_{2\rightarrow1})$ \\
$T_{ 4}$ & $H(W_4)$ \\
$T_{ 5}$ & $H(S_{3\rightarrow2},W_4)$ \\
$T_{ 6}$ & $H(S_{3\rightarrow1},S_{2\rightarrow1},W_4,W_1)$ \\
$T_{ 7}$ & $H(S_{4\rightarrow2},S_{4\rightarrow1},S_{3\rightarrow2},S_{3\rightarrow1},W_2,W_1)$ \\
$T_{ 8}$ & $H(S_{4\rightarrow3},M_1)$ \\
$T_{ 9}$ & $H(S_{4\rightarrow3},M_1,M_2)$ \\
$T_{10}$ & $H(S_{4\rightarrow2},S_{3\rightarrow2},M_1)$ \\
$T_{11}$ & $H(S_{4\rightarrow3},S_{3\rightarrow4},M_1)$ \\
$T_{12}$ & $H(S_{4\rightarrow2},S_{3\rightarrow4},M_1,M_2)$ \\
$T_{13}$ & $H(S_{4\rightarrow1},S_{3\rightarrow4},S_{2\rightarrow4},M_1,M_2)$ \\
$T_{14}$ & $H(S_{4\rightarrow1},S_{3\rightarrow4},S_{3\rightarrow2},S_{2\rightarrow4},M_1,M_2)$ \\
$T_{15}$ & $H(M_1)=B^+_1$ \\
$T_{16}$ & $H(M_1,M_2)=B^+_2$ \\
$T_{17}$ & $H(M_1,M_2,M_3)=B^+_3$ \\
\hline
\end{tabular}
\end{table}
\begin{table}
\centering
\caption{Proof by Tabulation of Proposition \ref{prop:MRR4_6}, with terms defined in Table \ref{table:MRR4_61}.\label{table:MRR4_62}}
\setlength{\tabcolsep}{3pt}
\begin{tabular}{|ccccccccccccccccc|}
\hline
$T_{ 1}$  &$T_{ 2}$  &$T_{ 3}$  &$T_{ 4}$  &$T_{ 5}$  &$T_{ 6}$  &$T_{ 7}$  &$T_{ 8}$  &$T_{ 9}$  &$T_{10}$  &$T_{11}$  &$T_{12}$  &$T_{13}$  &$T_{14}$  &$T_{15}$  &$T_{16}$  &$T_{17}$ 
\\
\hline
$ 20$     &$-10$     &          &          &          &          &          &          &          &          &          &          &          &          &          &          &           \\
$  5$     &$  5$     &$ -5$     &          &          &          &          &          &          &          &          &          &          &          &          &          &           \\
$  5$     &$  5$     &          &$ -5$     &          &          &          &          &          &          &          &          &          &          &          &          &           \\
          &          &          &$  1$     &$ -1$     &          &          &$  1$     &          &          &          &          &          &          &$ -1$     &          &           \\
          &          &$  2$     &          &$  2$     &$ -2$     &          &          &          &$ -2$     &          &          &          &          &          &          &           \\
          &          &          &$  2$     &          &$ -2$     &          &          &          &$  2$     &          &          &          &          &$ -2$     &          &           \\
          &          &          &          &          &$  1$     &          &          &$ -1$     &          &          &$  1$     &          &          &          &          &$ -1$     \\
          &          &          &$  1$     &$ -1$     &          &          &          &          &          &$  1$     &          &          &          &$ -1$     &          &           \\
          &          &$  1$     &$  1$     &          &$ -1$     &          &$ -1$     &          &          &          &          &          &          &          &          &           \\
          &          &          &          &          &$  1$     &$  1$     &          &          &          &          &          &          &$ -1$     &          &          &$ -1$     \\
          &          &          &          &          &$  1$     &          &          &          &          &          &          &$ -1$     &$  1$     &          &          &$ -1$     \\
          &          &$  2$     &          &          &          &$ -1$     &          &          &          &$ -1$     &          &          &          &          &          &           \\
          &          &          &          &          &$  1$     &          &          &          &          &          &$ -1$     &$  1$     &          &          &          &$ -1$     \\
          &          &          &          &          &$  1$     &          &          &$  1$     &          &          &          &          &          &          &$ -1$     &$ -1$     \\
          \hline\hline
$30$ &&&&&&&&&&&&&&$-4$&$-1$&$-5$         \\
\hline
\end{tabular}
\end{table}

\section{Concluding remarks}
\label{sec:conclusion}

We considered the problem of multilevel diversity coding with regeneration, which addresses the storage vs. repair-bandwidth tradeoff in distributed storage systems with heterogeneous reliability and access latency requirements. It was shown that for the minimum storage point on the optimal tradeoff curve, separate coding is sufficient, and there is no need to mix different contents. On the other hand, a complete characterization of the tradeoff region was provided for the case of four nodes, which reveals that mixing in general can strictly improve the overall tradeoff. 

Although we focused on the case $d=n-1$, some of the results can be generalized to $d<n-1$ straightforwardly, by recognizing that any MLD-R sytem with $d<n-1$ includes an MLD-R sub-system with $d=n\rq{}-1$. Particularly, the optimality of separate coding at the MSR point holds for $d<n-1$ as well. It is also worth mentioning that in a recent work \cite{Shao:16}, separate coding was shown to be also optimal at the MBR point, and thus the benefit of mixing only manifests in the intermediate tradeoff regime.


A notable feature of this work is that we further developed the computational approach in \cite{Tian:JSAC13} to identify and prove the converse theorems. As a result,  the converse proof was presented as tabulation without being translated into conventional form of proofs that are usually seen in information theory literature. It is our belief that this computational approach will  be able to play an even more significant role in future studies. To share our data with the research community, we have posted the computational results presented in this paper as part of the online collection of ``Solutions of Computed Information Theoretic Limits (SCITL)'' hosted at \cite{TianWebpage}, which we hope in the future can serve as a data depot for information-theoretic limits obtained through computational approaches. We are currently working toward extending the results obtained so far to more general parameters.  

There are several immediate research directions to follow. First, the code construction given for $n=4$ can be generalized to other parameters in a relatively straightforward manner, and we shall address this issue in a forthcoming work. Second, it is important to understand in general by mixing the contents how much improvement can be attained over separate coding. Finally, it may also be useful to consider the analogous requirement in the locally repairable code setting \cite{Huang:12,Papailiopoulos:ISIT12,Rawat:14}. 

\section*{Appendix: Proofs of Lemma \ref{prop:case3}, Lemma \ref{prop:case4}, Theorem \ref{theorem:case3} and Corollary \ref{coro:different}}
\begin{proof}[Proof of Lemma \ref{prop:case3}]
It is known that the normalized tradeoff region for the $(3,1,2)$ regenerating codes is given by the set of $(\bar{\alpha}_1,\bar{\beta}_1)$ pairs satisfying:
\begin{align}
\bar{\alpha}_1 \geq 1 \quad \mbox{and} \quad 2\bar{\beta}_1\geq 1
\end{align}
and the normalized tradeoff region for the $(3,2,2)$ regenerating codes is given by the set of $(\bar{\alpha}_2,\bar{\beta}_2)$ pairs satisfying:
\begin{align}
2\bar{\alpha}_2 \geq 1,\quad \bar{\alpha}_2+\bar{\beta}_2\geq 1,\quad \mbox{and} \;\; 3\bar{\beta}_2\geq 1.
\end{align}
Using Fourier-Motzkin elimination, it is straightforward to verify that the separate-coding normalized tradeoff region $\hat{\mathcal{R}}_3(\bar{B}_1,\bar{B}_2)=\{(\bar{\alpha}_1\bar{B}_1+\bar{\alpha}_2\bar{B}_2,\bar{\beta}_1\bar{B}_1+\bar{\beta}_2\bar{B}_2)\}$ is indeed given by Lemma~\ref{prop:case3}.
\end{proof}

\begin{proof}[Proof of Lemma \ref{prop:case4}]
It is known that the normalized tradeoff region for the $(4,1,3)$ regenerating codes is given by the set of $(\bar{\alpha}_1,\bar{\beta}_1)$ pairs satisfying:
\begin{align}
\bar{\alpha}_1\geq 1 \quad \mbox{and} \quad 3\bar{\beta}_1\geq 1 \label{eqn:413}
\end{align}
the normalized tradeoff region for the $(4,2,3)$ regenerating codes is given by the set of $(\bar{\alpha}_2,\bar{\beta}_2)$ pairs satisfying \cite{Dimakis:10}:
\begin{align}
2\bar{\alpha}_2\geq 1,\quad \bar{\alpha}_2+2\bar{\beta}_2\geq 1,\quad \mbox{and} \;\; 5\bar{\beta}_2\geq 1\label{eqn:423}
\end{align}
and the normalized tradeoff region for the $(4,3,3)$ regenerating codes is given by the set of $(\bar{\alpha}_3,\bar{\beta}_3)$ pairs satisfying \cite{Tian:JSAC13}:
\begin{align}
3\bar{\alpha}_3\geq 1, \quad 2\bar{\alpha}_3+\bar{\beta}_3\geq 1, \quad
4\bar{\alpha}_3+6\bar{\beta}_3\geq 3, \quad \mbox{and} \;\; 6\bar{\beta}_3\geq 1. \label{eqn:433}
\end{align}
However, unlike for $n=3$, using Fourier-Motzkin elimination to directly obtain a polyhedral description of $\hat{\mathcal{R}}_4(\bar{B}_1,\bar{B}_2,\bar{B}_3)$ is simply too time-consuming. Instead, denoting the set of $(\bar{\alpha},\bar{\beta})$ pairs constrained by the inequalities (\ref{eqn:case4_1})--(\ref{eqn:case4_5}) as $\tilde{\mathcal{R}}_4(\bar{B}_1,\bar{B}_2,\bar{B}_3)$, we shall show $\hat{\mathcal{R}}_4(\bar{B}_1,\bar{B}_2,\bar{B}_3) \subseteq \tilde{\mathcal{R}}_4(\bar{B}_1,\bar{B}_2,\bar{B}_3)$ and $\tilde{\mathcal{R}}_4(\bar{B}_1,\bar{B}_2,\bar{B}_3) \subseteq \hat{\mathcal{R}}_4(\bar{B}_1,\bar{B}_2,\bar{B}_3)$ separately.

To show that $\hat{\mathcal{R}}_4(\bar{B}_1,\bar{B}_2,\bar{B}_3) \subseteq \tilde{\mathcal{R}}_4(\bar{B}_1,\bar{B}_2,\bar{B}_3)$, we need to show that any $(\bar{\alpha},\bar{\beta})$ pair in $\hat{\mathcal{R}}_4(\bar{B}_1,\bar{B}_2,\bar{B}_3)$ must satisfy the inequalities (\ref{eqn:case4_1})--(\ref{eqn:case4_5}). Consider the inequality \eqref{eqn:case4_3} for example. For any $(\bar{\alpha},\bar{\beta}) \in \hat{\mathcal{R}}_4(\bar{B}_1,\bar{B}_2,\bar{B}_3)$, we have
\begin{align*}
4\bar{\alpha}+6\bar{\beta}
&=4(\bar{\alpha}_1\bar{B}_1+\bar{\alpha}_2\bar{B}_2+\bar{\alpha}_3\bar{B}_3)+6(\bar{\beta}_1\bar{B}_1+\bar{\beta}_2\bar{B}_2+\bar{\beta}_3\bar{B}_3)\\
&=(4\bar{\alpha}_1+6\bar{\beta}_1)\bar{B}_1+\left[\bar{\alpha}_2+3\left(\bar{\alpha}_2+2\bar{\beta}_2\right)\right]\bar{B}_2+(4\bar{\alpha}_3+6\bar{\beta}_3)\bar{B}_3\\
&\geq(4+2)\bar{B}_1+\left(\frac{1}{2}+3\right)\bar{B}_2+3\bar{B}_3\\
&=6\bar{B}_1+\frac{7}{2}\bar{B}_2+3\bar{B}_3,
\end{align*}
where the inequality above follows directly from the inequalities from (\ref{eqn:413})--(\ref{eqn:433}). The other four inequalities can be proved similarly; the details are omitted here. We thus conclude that $\hat{\mathcal{R}}_4(\bar{B}_1,\bar{B}_2,\bar{B}_3) \subseteq \tilde{\mathcal{R}}_4(\bar{B}_1,\bar{B}_2,\bar{B}_3)$.

To show that $\tilde{\mathcal{R}}_4(\bar{B}_1,\bar{B}_2,\bar{B}_3) \subseteq \hat{\mathcal{R}}_4(\bar{B}_1,\bar{B}_2,\bar{B}_3)$, first note that the characteristic cone of $\tilde{\mathcal{R}}_4(\bar{B}_1,\bar{B}_2,\bar{B}_3)$ is given by
$\{(\bar{\alpha},\bar{\beta}): \bar{\alpha} \geq 0,\; \bar{\beta}\geq0\}$. By the definition of $\hat{\mathcal{R}}_4(\bar{B}_1,\bar{B}_2,\bar{B}_3)$, any ray of $\tilde{\mathcal{R}}_4(\bar{B}_1,\bar{B}_2,\bar{B}_3)$ is also a ray of $\hat{\mathcal{R}}_4(\bar{B}_1,\bar{B}_2,\bar{B}_3)$. To examine the extreme points of $\tilde{\mathcal{R}}_4(\bar{B}_1,\bar{B}_2,\bar{B}_3)$, we can compute the intersections between any two inequalities (taken as equalities) from (\ref{eqn:case4_1})--(\ref{eqn:case4_5}). This yields a total of ten points, which are the possible extreme points of $\tilde{\mathcal{R}}_4(\bar{B}_1,\bar{B}_2,\bar{B}_3)$. However, some of them do not satisfy all the inequalities\footnote{More precisely, such a point violates certain inequalities unless certain components in $(\bar{B}_1,\bar{B}_2,\bar{B}_3)$ are zeros; however, under these degenerate conditions, it reduces to one of the points given in (\ref{eqn:point4a})--(\ref{eqn:point4d}).}, and after eliminating them, the possible extreme points of $\tilde{\mathcal{R}}_4(\bar{B}_1,\bar{B}_2,\bar{B}_3)$ are given by:
\begin{align}
\left(\bar{B}_1+\frac{\bar{B}_2}{2}+\frac{\bar{B}_3}{3},\frac{\bar{B}_1}{3}+\frac{\bar{B}_2}{4}+\frac{\bar{B}_3}{3}\right),\label{eqn:point4a}\\
\left(\bar{B}_1+\frac{\bar{B}_2}{2}+\frac{3\bar{B}_3}{8},\frac{\bar{B}_1}{3}+\frac{\bar{B}_2}{4}+\frac{\bar{B}_3}{4}\right),\label{eqn:point4b}\\
\left(\bar{B}_1+\frac{\bar{B}_2}{2}+\frac{\bar{B}_3}{2},\frac{\bar{B}_1}{3}+\frac{\bar{B}_2}{4}+\frac{\bar{B}_3}{6}\right),\label{eqn:point4c}\\
\mbox{and} \quad \left(\bar{B}_1+\frac{3\bar{B}_2}{5}+\frac{\bar{B}_3}{2},\frac{\bar{B}_1}{3}+\frac{\bar{B}_2}{5}+\frac{\bar{B}_3}{6}\right).\label{eqn:point4d}
\end{align}
Note from \eqref{eqn:413}--\eqref{eqn:433} that the extreme points of the normalized tradeoff rate regions for the $(4,1,3)$, $(4,2,3)$ and $(4,3,3)$ regenerating codes are given by:
\begin{align*}
(\bar{\alpha}_1,\bar{\beta}_1)&=\left(1,\frac{1}{3}\right),\\
(\bar{\alpha}_2,\bar{\beta}_2)&=\left(\frac{1}{2},\frac{1}{4}\right),\left(\frac{3}{5},\frac{1}{5}\right),\\
\mbox{and} \quad (\bar{\alpha}_3,\bar{\beta}_3)&=\left(\frac{1}{3},\frac{1}{3}\right),\left(\frac{3}{8},\frac{1}{4}\right),
\left(\frac{1}{2},\frac{1}{6}\right).
\end{align*}
Therefore,
\begin{itemize}
\item point \eqref{eqn:point4a} can be achieved by separate coding that uses $(4,1,3)$, $(4,2,3)$ and $(4,3,3)$ regenerating codes operating at normalized rate pairs $(\bar{\alpha}_1,\bar{\beta}_1)=\left(1,\frac{1}{3}\right)$, $(\bar{\alpha}_2,\bar{\beta}_2)=\left(\frac{1}{2},\frac{1}{4}\right)$ and $(\bar{\alpha}_3,\bar{\beta}_3)=\left(\frac{1}{3},\frac{1}{3}\right)$, respectively;
\item point \eqref{eqn:point4b} can be achieved by separate coding that uses $(4,1,3)$, $(4,2,3)$ and $(4,3,3)$ regenerating codes operating at normalized rate pairs $(\bar{\alpha}_1,\bar{\beta}_1)=\left(1,\frac{1}{3}\right)$, $(\bar{\alpha}_2,\bar{\beta}_2)=\left(\frac{1}{2},\frac{1}{4}\right)$ and $(\bar{\alpha}_3,\bar{\beta}_3)=\left(\frac{3}{8},\frac{1}{4}\right)$, respectively;
\item point \eqref{eqn:point4c} can be achieved by separate coding that uses $(4,1,3)$, $(4,2,3)$ and $(4,3,3)$ regenerating codes operating at normalized rate pairs $(\bar{\alpha}_1,\bar{\beta}_1)=\left(1,\frac{1}{3}\right)$, $(\bar{\alpha}_2,\bar{\beta}_2)=\left(\frac{1}{2},\frac{1}{4}\right)$ and $(\bar{\alpha}_3,\bar{\beta}_3)=\left(\frac{1}{2},\frac{1}{6}\right)$, respectively; and
\item point \eqref{eqn:point4d} can be achieved by separate coding that uses $(4,1,3)$, $(4,2,3)$ and $(4,3,3)$ regenerating codes operating at normalized rate pairs $(\bar{\alpha}_1,\bar{\beta}_1)=\left(1,\frac{1}{3}\right)$, $(\bar{\alpha}_2,\bar{\beta}_2)=\left(\frac{3}{5},\frac{1}{5}\right)$ and $(\bar{\alpha}_3,\bar{\beta}_3)=\left(\frac{1}{2},\frac{1}{6}\right)$, respectively.
\end{itemize}
We thus conclude that $\tilde{\mathcal{R}}_4(\bar{B}_1,\bar{B}_2,\bar{B}_3) \subseteq \hat{\mathcal{R}}_4(\bar{B}_1,\bar{B}_2,\bar{B}_3)$, completing the proof of Lemma~\ref{prop:case4}.
\end{proof}

\begin{proof}[Converse Proof of Theorem \ref{theorem:case3}]
To establish the converse of Theorem \ref{theorem:case3}, we shall prove that every normalized rate pair $(\bar{\alpha},\bar{\beta}) \in \mathcal{R}_3(\bar{B}_1,\bar{B}_2)$ must satisfy the inequalities from \eqref{eq:case3}. The inequality $\bar{\alpha} \geq \bar{B}_1+\frac{\bar{B}_2}{2}$ holds even without the regeneration requirement\cite{RocheYeungHau:97}, and the inequality $\bar{\alpha}+\bar{\beta} \geq \frac{3\bar{B}_1}{2}+\bar{B}_2$ follows directly from Theorem \ref{theorem:msp} by setting $n=3$. It remains to prove that the inequality $\bar{\beta} \geq \frac{\bar{B}_1}{2}+\frac{\bar{B}_2}{3}$ is true, which can be shown as follows.

First note that the repair bandwidth $\beta$ can be bounded from below as follows:
\begin{align}
\beta&\geq \frac{1}{2}[H(S_{1\rightarrow3})+H(S_{2\rightarrow3})]\nonumber\\
&\geq \frac{1}{2}H(S_{1\rightarrow3},S_{2\rightarrow3})\nonumber\\
&\stackrel{(a)}{=}\frac{1}{2}H(S_{1\rightarrow3},S_{2\rightarrow3},W_3,M_1)\nonumber\\
&\geq \frac{1}{2}B_1+\frac{1}{2}H(S_{1\rightarrow3},S_{2\rightarrow3},W_3|M_1)\label{eq:tpf1}
\end{align}
where $(a)$ is due to the fact that the data stored at node three $W_3$ can be regenerated from the helper messages $S_{1\rightarrow3}$ and $S_{2\rightarrow3}$. To proceed, we can further bound the second term on the right-hand side of \eqref{eq:tpf1} as follows:
\begin{align}
&H(S_{1\rightarrow3},S_{2\rightarrow3},W_3|M_1)\nonumber\\
&\stackrel{(a)}{=}H(S_{1\rightarrow3},S_{2\rightarrow3},W_3,S_{3\rightarrow1},S_{3\rightarrow2}|M_1)\nonumber\\
&\geq H(S_{1\rightarrow3},S_{2\rightarrow3},S_{3\rightarrow1},S_{3\rightarrow2}|M_1)\nonumber\\
&\stackrel{(s)}{=}\frac{1}{3}[H(S_{1\rightarrow3},S_{2\rightarrow3},S_{3\rightarrow1},S_{3\rightarrow2}|M_1)
+H(S_{1\rightarrow2},S_{3\rightarrow2},S_{2\rightarrow1},S_{2\rightarrow3}|M_1)\nonumber\\
&\qquad+H(S_{3\rightarrow1},S_{2\rightarrow1},S_{1\rightarrow2},S_{1\rightarrow3}|M_1)]\nonumber\\
&\stackrel{(b)}{\geq} \frac{1}{3}[H(S_{1\rightarrow3},S_{2\rightarrow3},S_{3\rightarrow1},S_{3\rightarrow2},S_{1\rightarrow2},S_{2\rightarrow1}|M_1)\nonumber\\
&\qquad+H(S_{3\rightarrow2},S_{2\rightarrow3}|M_1)+H(S_{3\rightarrow1},S_{2\rightarrow1},S_{1\rightarrow2},S_{1\rightarrow3}|M_1)]\nonumber\\
&\stackrel{(c)}{\geq}\frac{1}{3}[B_2+H(S_{3\rightarrow2},S_{2\rightarrow3}|M_1)+H(S_{3\rightarrow1},S_{2\rightarrow1},S_{1\rightarrow2},S_{1\rightarrow3}|M_1)]\nonumber\\
&\geq \frac{1}{3}[B_2+H(S_{3\rightarrow2},S_{2\rightarrow3},S_{3\rightarrow1},S_{2\rightarrow1},S_{1\rightarrow2},S_{1\rightarrow3}|M_1)]\nonumber\\
&\geq \frac{1}{3}[B_2+H(M_2|M_1)]\nonumber\\
&= \frac{2B_2}{3},\label{eq:tpf2}
\end{align}
where $(a)$ is due to the fact that the helper messages $S_{3\rightarrow1},S_{3\rightarrow2}$ are functions of $W_3$, $(b)$ follows from the submodularity of entropy function, and $(c)$ is because from $(S_{1\rightarrow3},S_{2\rightarrow3},S_{3\rightarrow1},S_{2\rightarrow1})$ we can regenerate $(W_1,W_3)$ and subsequently decode $M_2$. Substituting \eqref{eq:tpf2} into \eqref{eq:tpf1} gives
\begin{align}
\beta\geq \frac{1}{2}B_1+\frac{1}{3}B_2.
\end{align}
Normalizing both sides by $B_1+B_2$ completes the proof of $\bar{\beta} \geq \frac{\bar{B}_1}{2}+\frac{\bar{B}_2}{3}$ and hence the converse theorem.
\end{proof}

\begin{proof}[Proof of Corollary \ref{coro:different}]

Let us first show that when $\bar{B}_2\bar{B}_3=0$, we have $\mathcal{R}_4(\bar{B}_1,\bar{B}_2,\bar{B}_3)=\hat{\mathcal{R}}_4(\bar{B}_1,\bar{B}_2,\bar{B}_3)$. Since we have $\hat{\mathcal{R}}_4(\bar{B}_1,\bar{B}_2,\bar{B}_3) \subseteq \mathcal{R}_4(\bar{B}_1,\bar{B}_2,\bar{B}_3)$ {\em a priori}, we only need to show that $\mathcal{R}_4(\bar{B}_1,\bar{B}_2,\bar{B}_3) \subseteq \hat{\mathcal{R}}_4(\bar{B}_1,\bar{B}_2,\bar{B}_3)$. Further note that the inequality (\ref{eqn:case4_3}) is the only one from the set of inequalities (\ref{eqn:case4_1})--(\ref{eqn:case4_5}) that is not shared by the inequalities from the set of inequalities (\ref{eqn:rateregion4_1})--(\ref{eqn:rateregion4_6}), so we only need to show that any normalized rate pair $(\bar{\alpha},\bar{\beta}) \in \mathcal{R}_4(\bar{B}_1,\bar{B}_2,\bar{B}_3)$ must satisfy the inequality (\ref{eqn:case4_3}) when $\bar{B}_2\bar{B}_3=0$.

Note that when $\bar{B}_2=0$, (\ref{eqn:case4_3}) follows directly from (\ref{eqn:rateregion4_4}). On the other hand, when $\bar{B}_3=0$, from (\ref{eqn:rateregion4_2}) and (\ref{eqn:rateregion4_5}) we have
\begin{align} 
4\bar{\alpha}+6\bar{\beta}=\frac{2}{3}(2\bar{\alpha}+\bar{\beta})+\frac{8}{3}(\bar{\alpha}+2\bar{\beta})\geq 6\bar{B}_1+\frac{7}{2}\bar{B}_2
\end{align}
which is (\ref{eqn:case4_3}) when $\bar{B}_3=0$. This proves the ``if" part of the corollary.

To prove the ``only if" part, we shall assume that $\mathcal{R}_4(\bar{B}_1,\bar{B}_2,\bar{B}_3)=\hat{\mathcal{R}}_4(\bar{B}_1,\bar{B}_2,\bar{B}_3)$ and $\bar{B}_2\neq 0$. Note that when $\bar{B}_2\neq 0$, the inequality (\ref{eqn:case4_3}) does not follow directly from the inequality (\ref{eqn:rateregion4_4}). Since these two inequalities are ``parallel", so neither can be active within their respective groups of inequalities. Now consider the normalized rate pair
\begin{align*}
(\bar{\alpha},\bar{\beta}) = \left(\bar{B}_1+\frac{\bar{B}_2}{2}+\frac{7\bar{B}_3}{18},\frac{\bar{B}_1}{3}+\frac{\bar{B}_2}{4}+\frac{2\bar{B}_3}{9}\right).
\end{align*}
It is straightforward to verify that the above point satisfies the inequalities (\ref{eqn:case4_1}), (\ref{eqn:case4_2}), (\ref{eqn:case4_4}) and (\ref{eqn:case4_5}). Since the inequality (\ref{eqn:case4_3}) must be inactive within its group, the above point must satisfy the inequality (\ref{eqn:case4_3}) as well, which immediately implies that $\bar{B}_3=0$. We thus conclude that when $\mathcal{R}_4(\bar{B}_1,\bar{B}_2,\bar{B}_3)=\hat{\mathcal{R}}_4(\bar{B}_1,\bar{B}_2,\bar{B}_3)$, we must have $\bar{B}_2\bar{B}_3=0$. This completes the proof of the ``only if" part of the corollary.
\end{proof}

\section*{Acknowledgment}

The authors wish to thank one of the reviewers for pointing out an inaccurate statement in Corollary \ref{coro:different} in an earlier version of this paper. 

\bibliographystyle{IEEEbib}

\begin{thebibliography}{10}

\bibitem{Walsh:09}
J. M. Walsh, S. Weber, and C. wa Maina, \lq\lq{}Optimal rate delay tradeoffs and delay mitigating codes for multipath routed and network coded networks,\rq\rq{}
\newblock {\em IEEE Trans. Information Theory}, vol. 55, no. 12, pp. 5491--5510, Dec. 2009. 


\bibitem{Huang:12:ISIT}
L. Huang, S. Pawar, H. Zhang, and K. Ramchandran, ``Codes can reduce queueing delay in data centers,'' 
\newblock in {\em Proceedings 2012 IEEE International Symposium on Information Theory}, Cambridge, MA, USA, Jul. 2012, pp. 2766--2770.



\bibitem{Shah:14}
N. B. Shah,  Kangwook Lee, and K. Ramchandran, ``The MDS queue: Analysing the latency performance of erasure codes," 
\newblock in {\em Proceedings 2014 IEEE International Symposium on Information Theory}, Honolulu, HI, USA, Jun.-Jul. 2014, pp. 861--865.

\bibitem{RocheYeungHau:97}
J.~R. Roche, R.~W. Yeung, and K.~P.~Hau,
\newblock ``Symmetrical multilevel diversity coding,''
\newblock {\em IEEE Trans. Information Theory}, vol. 43, no. 5, pp. 1059--1064, May 1997.

\bibitem{YeungZhang:99}
R.~W. Yeung and Z. Zhang,
\newblock ``On symmetrical multilevel diversity coding,''
\newblock {\em IEEE Trans. Information Theory}, vol. 45, no. 2, pp. 609--621, Mar. 1999.


\bibitem{Dimakis:10}
A. G. Dimakis, P. B. Godfrey, Y. Wu, M. Wainwright and K. Ramchandran,
\newblock ``Network coding for distributed storage systems,''
\newblock {\em IEEE Trans. Information Theory}, vol. 56, no. 9, pp. 4539--4551, Sep. 2010.


\bibitem{Dimakis:11}
A. G. Dimakis, K. Ramchandran, Y. Wu, C. Suh,
\newblock ``A survey on network codes for distributed storage,''
\newblock {\em Proceedings of the IEEE}, vol. 99, no. 3, pp. 476--489, Mar. 2011.


\bibitem{RashmiShah:12:1}
N. B. Shah, K. V. Rashmi, P. V. Kumar and K. Ramchandran,
\newblock \lq\lq{}Distributed storage codes with repair-by-transfer and non-achievability of interior points on the storage-bandwidth tradeoff,\rq\rq{}
\newblock {\em IEEE Trans. Information Theory}, vol. 58, no. 3, pp. 1837--1852, Mar. 2012. 
 

\bibitem{RashmiShah:11}
K. V. Rashmi, N. B. Shah, and P. V. Kumar,
\newblock \lq\lq{}Optimal exact-regenerating codes for distributed storage at the MSR and MBR points via a product-matrix construction,\rq\rq{}
\newblock {\em IEEE Trans. Information Theory}, vol. 57, no. 8, pp. 5227--5239, Aug. 2011. 

\bibitem{Cadambe:11}
V. Cadambe, S. Jafar, H. Maleki, K. Ramchandran and C. Suh,
\newblock \lq\lq{}Asymptotic interference alignment for optimal repair of MDS codes in distributed storage,\rq\rq{}
\newblock {\em IEEE Trans. Information Theory}, vol. 59, no. 5, pp. 2974--2987, May 2013.

\bibitem{Tamo:13}
I. Tamo, Z. Y. Wang and J. Bruck,
\newblock \lq\lq{}Zigzag codes: MDS array codes with optimal rebuilding,\rq\rq{}
\newblock {\em IEEE Trans. Information Theory}, vol. 59, no. 3, pp. 1597--1616, Mar. 2013.

\bibitem{Papailiopoulos:13}
D. S. Papailiopoulos,  A. G. Dimakis and V. R. Cadambe,
\newblock \lq\lq{}Repair optimal erasure codes through Hadamard designs,\rq\rq{}
\newblock {\em IEEE Trans. Information Theory}, vol. 59, no. 5, pp. 3021--3037, May 2013.

\bibitem{Tian:JSAC13}
C. Tian,
\newblock \lq\lq{}Characterizing the rate region of the $(4,3,3)$ exact-repair regenerating codes,\rq\rq{}
\newblock {\em IEEE Journal on Selected Areas of Communications}, vol. 32, no. 5, pp. 967--975, May 2014.

\bibitem{Han-IC78} T.~S.~Han, ``Nonnegative entropy measures of multivariate symmetric correlations," \emph{Information and Control}, vol.~36, no.~2, pp.~133--156, Feb.~1978.

\bibitem{TianWebpage}
Solutions of Computed  Information Theoretic Limits (SCITL),
\newblock \url{http://web.eecs.utk.edu/~ctian1/SCITL.html}.

\bibitem{Shao:16}
S. Shao, T. Liu, and C. Tian, 
\newblock \lq\lq{}Multilevel diversity coding with regeneration: separate coding achieves the MBR point,\rq\rq{}
\newblock in {\em Proceedings 50th Annual Conference on Information Sciences and Systems}, Princeton, NJ, Mar. 2016.


\bibitem{Huang:12}
C. Huang, H. Simitci, Y. K. Xu, A. Ogus, B. Calder, P. Gopalan, J. Li, and S. Yekhanin,
\newblock \lq\lq{}Erasure coding in Windows Azure storage,\rq\rq{}
\newblock in {\em Proceedings 2012 USENIX Annual Technical Conference}, Boston, MA, USA, Jun.~2012, pp. 15--26.

\bibitem{Papailiopoulos:ISIT12}
D. S. Papailiopoulos and A. Dimakis,
\newblock \lq\lq{}Locally repairable codes,\rq\rq{}
\newblock {\em IEEE Trans. on Information Theory}, vol.~60, no.~9, pp. 5843--5855, Sep. 2014.


\bibitem{Rawat:14}
A. S. Rawat, O. O. Koyluoglu, N. Silberstein and S. Vishwanath, 
\newblock \lq\lq{}Optimal locally repairable and secure codes for distributed storage systems,\rq\rq{}
\newblock {\em IEEE Trans. on Information Theory}, vol.~60, no.~1, pp.~212--236, Jan. 2014.




\end{thebibliography}

\end{document}